\documentclass[12pt]{amsart}

\usepackage{amsmath,amsthm,amsfonts,amssymb,multicol}
\usepackage{mathrsfs}
\usepackage{dsfont}
\usepackage{graphicx}
\usepackage{booktabs}
\usepackage{array}
\usepackage{fullpage}
\usepackage{subfigure}
\usepackage{enumerate}
\usepackage{url}

\usepackage[english]{babel}
\usepackage{verbatim} 
\usepackage{enumitem}
\usepackage{bbm} 
\usepackage[normalem]{ulem} 
\usepackage{bm}
\usepackage{mathtools}
\usepackage[pdftex]{hyperref}
\usepackage{graphicx}
\hypersetup{pdfpagemode=}

\usepackage{todonotes}

\newcommand{\beq}{\begin{equation}}
\newcommand{\eeq}{\end{equation}}
\newcommand{\beqs}{\begin{equation*}}
\newcommand{\eeqs}{\end{equation*}}

\renewcommand{\d}{\mathrm{d}} 

\renewcommand{\Im}{\mathrm{Im}}

\newcommand{\eps}{\veps}
\newcommand{\veps}{\varepsilon}

\newcommand{\Z}{\mathbb{Z}}
\newcommand{\N}{\mathbb{N}}
\newcommand{\C}{\mathbb{C}}
\newcommand{\R}{\mathbb{R}}

\newcommand{\cB}{\mathcal{B}}

\newcommand{\cR}{\mathcal{R}}

\newcommand{\sH}{\mathscr{H}}

\newcommand{\ue}{\mathrm{e}}
\newcommand{\ui}{\mathrm{i}}

\newcommand{\1}{\mathds{1}}

\newcommand{\supp}{\,\text{supp}}

\newcommand{\ds}{\,\d s}

\newcommand{\Tr}[1]{\textnormal{Tr}\left[#1\right]}	
\newcommand{\tr}[1]{\textnormal{tr}\left[#1\right]}	
\newcommand\myfrac[2]{\genfrac{}{}{0pt}{}{#1}{#2}}

\newtheorem{thm}{Theorem}[section]
\newtheorem{lm}[thm]{Lemma}
\newtheorem{cor}[thm]{Corollary}
\newtheorem{prop}[thm]{Proposition}

\newtheorem{rmk}[thm]{Remark}

\theoremstyle{definition}
\newtheorem{defn}[thm]{Definition}

\newtheorem{exam}[thm]{Example}

\theoremstyle{remark}

\usepackage{color}

\def\dotuline{\bgroup
  \ifdim\ULdepth=\maxdimen  
   \settodepth\ULdepth{(j}\advance\ULdepth.4pt\fi
  \markoverwith{\begingroup
  \advance\ULdepth0.08ex
  \lower\ULdepth\hbox{\kern.15em .\kern.1em}%
  \endgroup}\ULon}

\def\dashuline{\bgroup
  \ifdim\ULdepth=\maxdimen  
   \settodepth\ULdepth{(j}\advance\ULdepth.4pt\fi
  \markoverwith{\kern.15em
  \vtop{\kern\ULdepth \hrule width .3em}%
  \kern.15em}\ULon}
\allowdisplaybreaks

\title{Dynamics and equilibrium states of infinite systems of lattice bosons}

\author{Andreas Deuchert}
\address{Andreas Deuchert, 448 McBryde Hall, Virginia Tech,
225 Stanger Street,
Blacksburg, VA 24061-1026}
\email{andreas.deuchert@vt.edu}

\author{ Jonas Lampart}
\address{Jonas Lampart, CNRS \& LICB, UMR 6303 Université Bourgogne Europe, 9 Av. A. Savary, 21078 Dijon}
\email{jonas.lampart@u-bourgogne.fr}

\author{Marius Lemm}
\address{Marius Lemm, University of Tübingen,
Department of Mathematics,
Auf der Morgenstelle 10,
72076 Tübingen}
\email{marius.lemm@uni-tuebingen.de}

\begin{document}

\begin{abstract}
We consider the dynamics of systems of lattice bosons with infinitely many degrees of freedom. We show that their dynamics defines a group of automorphisms on a $C^*$--algebra introduced by Buchholz,
which extends the resolvent algebra of local field operators. For states that admit uniform bounds on moments of the local particle number, we derive propagation bounds of Lieb--Robinson type. Using these bounds, we show that the dynamics of local observables gives rise to a strongly continuous unitary group in the GNS representation. Moreover, accumulation points of finite-volume Gibbs states satisfy the KMS condition with respect to this group. This, in particular, proves the existence of KMS states.
\end{abstract}

\maketitle
\tableofcontents

\section{Introduction}

In this article, we are interested in the dynamics of systems of lattice bosons. A paradigmatic example of such a system is the Bose--Hubbard model on $\mathbb Z^d$, defined by the Hamiltonian
\begin{equation}
 H=- \sum_{\myfrac{x,y\in \Z^d}{ d(x,y)=1}} (a_x^*a_y + a_y^* a_x) + U \sum_{x\in \Z^d}  N_x (N_x-1).
 \label{eq:BoseHubbardHamiltonian}
\end{equation}
Here, $a^*_x$ and $a_x$ are the usual creation and annihilation operators of a particle at site $x \in \mathbb{Z}^{d}$ and $N_x = a_x^* a_x$ is the associated number operator. The first term of the Hamiltonian describes nearest-neighbor hopping, and the second term corresponds to onsite boson-boson interaction of strength $U > 0$. The Bose--Hubbard model is expected to exhibit an interesting phase diagram~\cite{Fisher1989, Freericks1994, Greiner2002}, including Bose--Einstein condensed (BEC) and Mott insulator phases. A proof of BEC for the model with hard core repulsion ($U=+\infty$) and density equal to $1/2$ was given in 1978 in the seminal article\footnote{The proof has been provided for a class of quantum spin systems including the quantum $xy$ model, which can be mapped to the hard-core Bose gas at half-filling, see e.g. \cite{ALSSY2004}.} \cite{DLS1978} using reflection positivity. The existence of a version of the BEC--Mott insulator transition was established in 2004 in \cite{ALSSY2004}.  

We are interested in studying bosonic lattice systems with infinitely many degrees of freedom, focusing primarily on the existence of equilibrium states and the study of near-equilibrium dynamics. The natural mathematical framework for describing such systems is provided by quasi-local $C^*$-algebras, states on these algebras, and automorphism groups. However, defining an interacting bosonic dynamics in this framework is a notoriously difficult problem because \textit{infinitely} many particles can accumulate in a \textit{finite} region of space.

One approach that has been developed to treat such systems is based on the study of the thermodynamic limit of time-dependent Green functions, $\gamma_{\Lambda}(\tau^{\Lambda}_t(A)B)$, for a local Gibbs state $\gamma_{\Lambda}$, $\Lambda \subset \Z^{d}$, a local time evolution $\tau^{\Lambda}_t$, and two local observables $A,B$, see~\cite[Ch.~6.3.4]{BraRob2}.
In \cite{Park1984, Park1985}, Park proved bounds on exponential moments of the local number operators for a system of bosons moving in $\mathbb{R}^d$. 
Using these bounds and a compactness argument, he concluded that Green's functions have (possibly non-unique) thermodynamic limits. Using~\cite[Thm.~6.3.27]{BraRob2}, one can reconstruct a state $\gamma$ on a quasi-local algebra and a unitary group acting on a representing Hilbert space of the algebra related to $\gamma$. Moreover, the limiting dynamics satisfy a form of the Kubo--Martin--Schwinger (KMS)  equilibrium condition.
However, this construction is not completely satisfactory since the Hilbert space can be very large, containing for every time an isomorphic copy of the Hilbert space associated with $t=0$.
In~\cite{PY1994, PY1995} a similar approach was used to study a system of oscillators on a lattice interacting via superstable interactions, see also~\cite{olivieri1993}.

An alternative approach to bosonic systems with infinitely many degrees of freedom was pioneered by Buchholz and Grundling~\cite{BG2007, BG2008, BG2013, BG2015} and developed further in~\cite{Buchholz2014, Buchholz2017, buchholz2018,Buchholz2020, BuchholzBahns2021, Buchholz2022a,Buchholz2022b, BuNu2023, BuchholzYngvason2024, NuSi2025, KoVdV2025,NuVdV2026}.
This approach is based on the $C^*$-algebra generated by resolvents of local field operators (the resolvent algebra).
For finitely many degrees of freedom, this algebra is invariant under time-evolutions generated by Hamiltonians with pair interactions~\cite{BG2008}, and for bosons moving in $\mathbb{R}^d$ this holds for a natural extension $\mathcal{B}$ of the algebra~\cite{buchholz2018}.
That is, for any $A \in \mathcal{B}$ one has 
\begin{equation}
    \tau_t(A) = \ue^{\mathrm{i} H t} A \ue^{-\mathrm{i} H t}  \in \mathcal{B}.
    \label{eq:defDynamicsOnB}
\end{equation}
The dynamics generated by $H$ can therefore be defined on the $C^*$-algebra $\mathcal{B}$ as a group of $*$-automorphisms. As one would expect from the bosonic nature of the particles, this dynamics exhibits only very weak locality properties and it is not even weakly continuous in time (see Example~\ref{ex:continuity} below). 
The framework of the resolvent algebra has been used in \cite{BuchholzBahns2021,Buchholz2022a,Buchholz2022b,BuchholzYngvason2024} to provide an alternative characterization of Bose–Einstein condensation in systems with infinitely many degrees of freedom and to construct equilibrium states for non-interacting models. Methods developed in \cite{buchholz2018} have also been applied in the fermionic setting~\cite{Siebert2024} to identify an extension of the CAR algebra that remains invariant under the dynamics generated by a Hamiltonian with a pair interaction.

As already hinted at above, the potential nonlocality of the dynamics on $\mathcal{B}$ stems from the fact that there exist states on the algebra describing an infinite number of particles localized in a finite region of space. Such configurations have infinite local energy, and large numbers of particles may therefore cover long distances in arbitrarily short times. 
However, if one restricts attention to reasonable states with uniform bounds on the local particle number,
such pathological behavior is not expected.
The main technical tool we will use to show this are state-dependent Lieb--Robinson-type bounds.

Lieb--Robinson bounds were first proved for quantum spin systems in~\cite{lieb1972finite}. In that case, as well as for lattice fermions, they control the dynamical spreading of local observables under the Heisenberg time evolution in operator norm. Using these bounds, one can approximate the time evolution of a local observable by evolution under a local Hamiltonian, and thereby prove the existence of the thermodynamic limit of the dynamics on an algebra of quasi-local observables~\cite{BraRob2,nachtergaele2006propagation,nachtergaele2010existence,gebert2020lieb,hinrichs2024lieb}. Sparked by discoveries of  Hastings \cite{hastings2004lieb,hastings2005quasiadiabatic,hastings2007area}, Lieb--Robinson bounds have become 
decisive analytical tools for a variety of problems in mathematical condensed-matter physics \cite{hastings2004lieb,hastings2005quasiadiabatic,nachtergaele2006lieb,bravyi2010topological,bachmann2012automorphic,bachmann2018adiabatic} and quantum information theory \cite{bravyi2006lieb,hastings2007area,epstein2017quantum, kliesch2014lieb,woods2015simulating,haah2021quantum}.
However, the standard proofs of Lieb–Robinson bounds break down for systems with unbounded local interactions. For certain perturbations of harmonic oscillator systems, Lieb–Robinson bounds can be proven \cite{nachtergaele2009lieb} 
and these bounds can then be used to obtain the thermodynamic limit of the dynamics \cite{nachtergaele2010existence, NuVdV2026}. In contrast, for the Bose–Hubbard model, such state-independent bounds are not expected to hold. The situation becomes even more difficult when considering particles in $\mathbb{R}^d$ (instead of on a lattice),
where an ultraviolet cutoff is typically required even for fermions~\cite{gebert2020lieb, hinrichs2024lieb, bachmann2024lieb}. Following prior progress on state-dependent Lieb--Robinson bounds for Bose-Hubbard type Hamiltonians  in~\cite{schuch2011information,kuwahara2021lieb,faupin2022lieb,yin2022finite}, Kuwahara, Vu, and Saito~\cite{KVS2024} recently proposed an alternative approach that restricts attention to initial states with well-behaved local particle number; see also \cite{kuwahara2024enhanced}. This allows them to approximate the generator of the dynamics locally by a bounded Hamiltonian, and to prove Lieb–Robinson-type bounds. 

In this article, we use state-dependent Lieb--Robinson-type bounds, in the spirit of \cite{KVS2024}, to prove new results about the lattice version of the infinite-volume dynamics defined on Buchholz’s extension $\mathcal{B}$ of the resolvent algebra. We first construct the dynamics $\tau_t$ on $\mathcal{B}$ for Bose--Hubbard-type Hamiltonians. 
Then, we derive Lieb--Robinson-type bounds for the dynamics of states that satisfy uniform bounds on moments of the local particle number. Our bounds are weaker than those obtained in \cite{KVS2024} with respect to the relation between spatial and temporal distance. This is sufficient for the purposes of studying infinite-volume dynamics, where the bounds are applied for fixed times with distance going to infinity. Thanks to this different focus, we are able to give a substantially simpler, self-contained proof.
Using the existence of the dynamics $\tau_t$ on $\mathcal{B}$ and our state-dependent Lieb--Robinson bounds, we prove two new operator-algebraic results for systems of lattice bosons with infinitely many degrees of freedom that would not have been possible without these two new ingredients:
\begin{enumerate}[label=(\alph*)]
    \item The dynamics of quasi-local observables can be implemented by a strongly continuous unitary group on the GNS Hilbert space associated with a $\tau_t$-invariant state that satisfies the prescribed moment bounds. Furthermore, this global dynamics can be approximated by the corresponding local dynamics.
    \item 
    Any accumulation point of finite-volume Gibbs states with uniform local particle-number bounds satisfies
    the KMS condition. This, in particular, proves the existence of KMS states for the Bose--Hubbard model.
 \end{enumerate}

Our article is organized as follows. In Section~\ref{sect:setup}, we introduce the resolvent algebra and Buchholz’s extension $\mathcal{B}$, and discuss some of their properties. The dynamics $\tau_t$ on $\mathcal{B}$ is constructed in Section~\ref{sec:existenceOfDynamics}. Section~\ref{sect:approx} is devoted to the proof of Lieb–Robinson-type bounds, which are then applied to define a strongly continuous unitary group on the GNS Hilbert space associated with a $\tau_t$-invariant state. Finally, in Section~\ref{sect:KMS}, we prove the existence of KMS states. Sections~\ref{sec:existenceOfDynamics}--\ref{sect:KMS} begin with the statement of the main results, followed by their proofs in the remainder of the section.

\section{The framework}\label{sect:setup}
In order to allow for different boundary conditions and geometries, we will consider a slightly more general setting than in the introduction.
As one-particle configuration space we choose a graph $\Gamma$
and endow it with the geodesic distance $d(x,y)$. This is the number of edges on the shortest path between $x$ and $y$ in $\Gamma$, i.e., $d(x,y)=1$ if $x$ is joined to $y$ by an edge.
For a given subset $X \subseteq \Gamma$ and a number $\ell \geq 0$ we denote the $\ell$-enlargement of $X$ by
\begin{equation}
	X[\ell] = \{ x \in \Gamma \ | \ d(x,X) \leq \ell \},
	\label{eq:enlargement}
\end{equation}
and
the (exterior) boundary of $X$ as
\begin{equation}
	\partial X = X[1]\setminus X =\{ x \in X^c \ | \ d(x,X) = 1 \}.
	\label{eq:boundary}
\end{equation}
If the set $X$ consists of a single point $x$ we write $X[r] = x[r]$. We assume that $\Gamma$ is $d$-dimensional in the sense of the following definition. The prototypical example is $\mathbb{Z}^d$.

\begin{defn}
	\label{def:dDimGraph}
	We say a graph $\Gamma$ is $d$-dimensional, $d\geq 1$, with surface parameter $\sigma > 0$ if
	\begin{equation}
		\sup_{x \in \Gamma} \sup_{\ell \geq 0} \frac{| \partial(x[\ell]) |}{(\ell+1)^{d-1}} \leq \sigma
		\label{eq:dDimGraph}
	\end{equation}
	holds.  
\end{defn}
Note that this implies $|x[\ell]|\leq \sigma \ell^d+1$, and $|X[\ell]|\leq |X|(\sigma\ell^d+1)$ for all $X\Subset \Gamma$. Here, and in the following, we write $X\Subset \Gamma$ for a finite subset $X$ of $\Gamma$.

Our main concern in this article is the study of the dynamics of bosonic lattice gases with bounded local particle number but with infinitely many particles in the whole space. The appropriate mathematical formalism to study this problem is that of quasi-local algebras, which we introduce later. We will mostly work in a concrete realisation of these algebras on the bosonic Fock space, which is defined by
\begin{equation}
	\mathscr{F} = \mathscr{F}(\ell^2(\Gamma)) := \bigoplus_{n=0}^{\infty} \ell^2(\Gamma)^{\otimes_{\mathrm{s}}^n}.
	\label{eq:localFockSpace}
\end{equation}
Here, $\otimes_{\mathrm{s}}^n$ denotes the $n$-fold symmetric tensor product and $\ell^2(\Gamma)^0 := \mathbb{C}$. 
The Fock space satisfies the exponential identity 
$\mathscr{F}(\mathcal{H}_1 \oplus \mathcal{H}_2) \cong \mathscr{F}(\mathcal{H}_1) \otimes \mathscr{F}(\mathcal{H}_2)$ with two Hilbert spaces $\mathcal{H}_1, \mathcal{H}_2$ and unitary equivalence denoted by $\cong$. In particular, for any set $\Lambda \subset \Gamma$ we have, denoting $\mathscr{F}(\ell^2(\Lambda)) =\mathscr{F}_\Lambda$,
\begin{equation}
    \mathscr{F} \cong \mathscr{F}_\Lambda \otimes \mathscr{F}_{\Lambda^{\mathrm{c}}},
    \label{eq:ExponentialIdentity}
\end{equation}
with $\Lambda^{\mathrm{c}}=\Gamma\setminus\Lambda$. We will, by a slight abuse of notation, denote operators and states on the two spaces by the same symbol. 

For $f\in \ell^2(\Gamma)$, we denote by $a^*(f)$ and $a(f)$ the usual, densely defined, bosonic creation and annihilation operators on $\mathscr{F}$. They satisfy the canonical commutation relations (CCR)
\begin{equation}
	[a(f),a^*(g)] = \langle f, g \rangle, \quad [a(f),a(g)] = 0 = [a^*(f),a^*(g)].
	\label{eq:CCR}
\end{equation}
For $x\in \Gamma$ let $\delta_x(y) = \delta_{x,y}$. By $a^*_x=a^*(\delta_x)$ and $a_x=a(\delta_x)$ we denote the creation and the annihilation operator of a particle at site $x$. Moreover, $N_x=a_x^* a_x$ and
\begin{equation}
	N = \sum_{x \in \Gamma} N_x
	\label{eq:numberOperator}
\end{equation} 
denote the self-adjoint number operators. The general form of the Hamiltonian with two-boson interaction  $v: \Gamma \times \Gamma \to \R$ on a subset $\Lambda\subseteq \Gamma$ reads
\begin{align}
	H_\Lambda &= -\sum_{\myfrac{x,y \in \Lambda : }{ d(x,y) = 1}} (a_x^* a_y + \mathrm{h.c.} ) + \sum_{x,y \in \Lambda} v(x,y) a_x^* a_y^* a_x a_y
	\label{eq:localHamiltonian} \\
	&=: \sum_{\myfrac{x,y \in \Lambda :} { d(x,y) = 1}} T_{xy} + \sum_{x,y \in \Lambda} V_{xy}
	=T_\Lambda+V_\Lambda \notag.
\end{align}
We assume that $v$ is bounded and satisfies $v(x,y)=v(y,x)$ and $v(x,y) \to 0$ if $d(x,y) \to \infty$, and we denote $H:=H_\Gamma$.
The Hamiltonian in $H$ is essentially self-adjoint on the linear subspace of $\mathscr{F}$ consisting of finite linear combinations of eigenvectors of the number operator $N$. We note that our results can likely be generalized to include position-dependent hopping terms and many-body interactions. For the sake of simplicity we prefer to work with the Hamiltonian in \eqref{eq:localHamiltonian}.
However, it is important for us that $H$ preserves the number of particles.

\subsection{Algebras}
We now introduce the relevant algebras for our discussion of the dynamics.
All of them are realised as subalgebras of $\mathcal{L}(\mathscr{F})$, the bounded operators on Fock space.
Accordingly, the time evolution $A\mapsto \ue^{\ui Ht}A\ue^{-\ui Ht}$ is always well defined, and the question whether this is an evolution on an algebra $\mathcal{A}\subset \mathcal{L}(\mathscr{F})$ is reduced to the question of invariance of $\mathcal{A}$ under the above unitary conjugation.

For $f \in \ell^2(\Gamma)$ we define the self-adjoint field operator $\phi(f)=a(f) + a^*(f)$.
The resolvent algebra $\mathcal{R}_X$, introduced in~\cite{BG2007, BG2008}, over a set $X\subseteq \Gamma$ is the $C^*$-algebra generated by the identity and the operators
\begin{equation}
	R(f;z) := (\phi(f) - z)^{-1} \quad \text{ with } \quad f \in \ell^2_\mathrm{fin}(X),\, \Im z\neq 0.
	\label{eq:resolvent}
\end{equation}
Here, $\ell^2_\mathrm{fin}(X) \subseteq \ell^2(X)$ denotes the linear subspace of all functions with finite support. We highlight that $X$ may be infinite. The algebra $\mathcal{R}_X$ is not separable~\cite[Thm. 5.3]{BG2008}.

We will restrict ourselves to observables that commute with the number operator, as does the Hamiltonian.
The subalgebra $\mathcal{R}^{\mathrm{inv}}_X\subseteq \mathcal{R}_X$ consists of those  $A\in \mathcal{R}_X$, for which $\ue^{\ui t N}A\ue^{-\ui t N}=A$ holds for all $t\in \R$. Note that the above conjugation is induced by the gauge transformation $f\mapsto \ue^{\ui t}f$ for $f \in \ell^2(\Gamma)$. For $X=\Gamma$, we denote $\mathcal{R}=\mathcal{R}_{\Gamma}$, $\mathcal{R}^{\mathrm{inv}}=\mathcal{R}^{\mathrm{inv}}_\Gamma$. In the following we refer to $\mathcal{R}^{\mathrm{inv}}_X$ as the gauge-invariant resolvent algebra. 

The Buchholz algebra, introduced in~\cite{buchholz2018} for the case of bosonic systems in $\mathbb{R}^d$,
is an extension of the gauge-invariant resolvent algebra. As we will show below, this algebra is left invariant by the dynamics generated by $H$.

\begin{defn}(Buchholz algebra)
For $X\subseteq \Gamma$, the Buchholz algebra $\mathcal{B}_X$ consists of all bounded operators $A \in \mathcal{L}(\mathscr{F})$, such that for any $n \in \mathbb{N}_0$ there exists $R \in \mathcal{R}^{\mathrm{inv}}_X$ with
	\begin{equation*}
		A \1_{N\leq n}  = R\1_{N\leq n}.
	\end{equation*}
	We denote $\mathcal{B}=\mathcal{B}_\Gamma$. 
\end{defn}
We highlight that the Buchholz algebra is a subalgebra of $\mathcal{L}(\mathscr{F})$. In particular, the norm on $\mathcal{B}_X$ is that on $\mathcal{L}(\mathscr{F})$. Using this, it is straightforward to check that $\mathcal{B}_X$ is a $C^*$-algebra. 

By definition, the difference between the algebras $\mathcal{B}$ and $\mathcal{R}^{\mathrm{inv}}$ lies only in the behavior for $n\to \infty$. For example, one easily checks that $\mathcal{R}^{\mathrm{inv}}$ contains the resolvents $(N_x+z)^{-1}$, and thus, by the Stone-Weierstrass theorem (since also $\1\in \mathcal{R}^{\mathrm{inv}}$), all functions $f(N_x)$ that admit a limit for $N_x\to \infty$. On the other hand, $\mathcal{B}$ clearly contains all operators of the form $f(N_x)$ with bounded functions $f$ that may not have a limit for $N_x \to \infty$. We refer to \cite[p.~961]{buchholz2018}, where this example has also been discussed. 

The structure of the above algebras can be made more explicit. Let $\mathscr{F}_n=\1_{N=n} \mathscr{F}$ be the closed linear subspace of $\mathscr{F}$ with exactly $n$ particles, and denote by $\mathcal{K}_n$ the algebra of compact operators on $\mathscr{F}_n$. Then, define the algebra
\begin{equation}
	\mathcal{C}_{m,n} = \underbrace{ \mathds{1}_{\ell^2(\Gamma)} \otimes_{\mathrm{s}} ... \otimes_\mathrm{s} \1_{\ell^2(\Gamma)}  }_{n-m \text{ times }} \otimes_{\mathrm{s}} \mathcal{K}_m
	\label{eq:Calgebra}
\end{equation}
acting on $\mathscr{F}_n$ with $n \geq m$. It consists of operators that act on $m$ particles in a compact way and on $n-m$ particle as the identity. By $\mathfrak{K}_n$ we denote the $C^*$-algebra generated by $\mathcal{C}_{m,n}$ with $0 \leq m \leq n$, which is simply the direct sum over $m=0,\dots, n$.
The following lemma has been proved in \cite[Lemma 3.3]{buchholz2018} for the case of particles moving in $\mathbb{R}^d$. Its proof in our setting is almost literally the same, and therefore omitted.

\begin{lm}\label{lem:structureResolventAlgebra}
For $n \in \mathbb{N}_0$ we have $\cR^\mathrm{inv}|_{\mathscr{F}_n} = \mathfrak{K}_n$.
\end{lm}

The restrictions of $A\in \cR^\mathrm{inv}$ to $\mathscr{F}_n$ are, however, not independent. Indeed, the restriction to $\mathscr{F}_n$ uniquely determines that to $\mathscr{F}_m$ for $m<n$ via the map
\begin{equation}\label{eq:defKappaN}
  \begin{aligned}
\kappa_n:\mathfrak{K}_n &\to \mathfrak{K}_{n-1} \\
K_n+ \sum_{m=0}^{n-1}\1_{\mathscr{F}_{n-m}} \otimes_{\mathrm{s}} K_m &\mapsto   n^{-1} K_{n-1} + n^{-1}\sum_{m=0}^{n-2} (n-m)   \1_{\mathscr{F}_{n-m-1}} \otimes_{\mathrm{s}}  K_m,
 \end{aligned}  
\end{equation}
which simply deletes the element of $\mathcal{K}_n$ and otherwise drops one factor of the identity, adjusting the prefactors to account for symmetrization. 
The relation $\kappa_n$ can be motivated by considering a product state of $n$ particles and sending one of them to infinity, cf.~\cite[Lem.~3.4]{buchholz2018}.
This provides a complete characterization of $\mathcal{B}$ as the inverse limit of $\mathcal{K}_n$, see~\cite[Thm.~3.5]{buchholz2018}.

\begin{prop}\label{prop:Buchholz}
The algebra $\mathcal{B}$ is given by all bounded sequences $n\mapsto A_n\in \mathfrak{K}_n$ that are consistent in the sense that $\kappa_{n}(A_n)=A_{n-1}$, endowed with the norm $\|A\|=\sup_{n\in \N_0} \|A_n\|$.
\end{prop}

The resolvent and Buchholz algebras come with a natural local structure, that is, for $X\subseteq Y$ we have the inclusion $\mathcal{R}_X\subseteq  \mathcal{R}_Y$, $\mathcal{B}_X\subseteq  \mathcal{B}_Y$.
An algebra  with such a geometric structure is called quasi-local if every element can be approximated in norm by elements of subalgebras associated with finite sets $X \Subset \Gamma$. The resolvent algebra $\mathcal{R}$ is quasi local, since it is generated by resolvents $R(f;z)$ with functions $f$ of finite support.
The Buchholz algebra, however, is not quasi-local. This is illustrated by the following example.

\begin{exam}\label{ex:R(f)}
Consider a function $f\in \ell^2(\Gamma)\setminus \ell^2_\mathrm{fin}(\Gamma)$ and the bounded operator $A=(1+ a^*(f) a(f))^{-1}$. Of course, $f$ can be approximated in the norm of $\ell^2(\Gamma)$ by a sequence $(f_k)_{k \in \N}$ of functions with finite support, and the corresponding operators $A_k=(1+ a^*(f_k) a(f_k))^{-1}$ are elements of $\mathcal{R}^\mathrm{inv}\subset \mathcal{R}$. The restrictions to the $n$-particle space converge in operator norm, $A_k\vert_{\mathscr{F}_n}\to A\vert_{\mathscr{F}_n}$ and thus $A\vert_{\mathscr{F}_n}\in \mathfrak{K}_n$ because $\mathfrak{K}_n$ is closed. Since the consistency relations also pass to the limit by continuity, we conclude that $A\in \mathcal{B}$. 
However, without the restriction, $A_k$ converges to $A$ only in the strong operator topology, and not in norm.
Indeed, for $X=\supp f_k$ and a vector $\psi= \phi \otimes \Omega_{X}$ (cf.~\eqref{eq:ExponentialIdentity}) with no particles inside $X$, $\langle \psi, A_k \psi \rangle$ equals one, while $\langle \psi, A \psi \rangle$ can be arbitrarily small.
The algebra $\mathcal{B}$ is quasi-local when restricted to the $n$-particle sector (compare~\cite[p.~970]{buchholz2018}), but the approximation by localized operators is not uniform in $n$.
\end{exam}

\subsection{States}

A state on a $C^*$-algebra (with identity)  $\mathcal{A}\subseteq \mathcal{L}(\mathscr{F})$ is a continuous, linear functional $\gamma :\mathcal{A}\to \C$ with $\gamma(A^* A) \geq 0$ for all $A \in \mathcal{A}$ and $\gamma(\1)=1$. A state $\gamma$ is called normal if  there exists a trace-class operator $\rho$ on $\mathscr{F}$ so that $\gamma(A)=\Tr{\rho A}$ (compare~\cite[Thm.~2.4.21]{BraRob1}).
Normal states have (almost-surely) a finite particle number, since the probability of having $n$ particles equals $\gamma(\1_{N=n})$ and we have $\sum_n \gamma(\1_{N=n})=\Tr{\rho}= 1$.
In view of Lemma~\ref{lem:loc_normal} below (see also \cite[Thm. 5.2.14]{BraRob2}), normality is essentially equivalent to $\gamma$ having finite particle number in this sense.

As mentioned before, we are not interested in states with a finite particle number but rather in states with infinitely many particles and a finite density. The algebras $\mathcal{R}$, $\mathcal{B}$ admit states with infinitely many particles locally and globally. For example, $\gamma(f(N_x))=\lim_{n\to \infty} f(n)$ defines a state on the closed subalgebra of $\mathcal{R}_{x}$ of functions of the particle number operator $N_x$ that have a limit. Using the Hahn--Banach theorem and the fact that any continuous linear functional $\phi$ on a unital $C^*$-algebra  with $\phi(\mathds{1}) = \Vert \phi \Vert$ is a state, see e.g. \cite[Proposition~2.3.11]{BraRob1}, $\gamma$ can be extended to a state on $\mathcal{R}$ and $\mathcal{B}$. In the following, we will exclude such states with a condition on the moments of the local particle number. For a state $\gamma$ on $\mathcal{B}$ and $x\in \Gamma$, $p> 0$, we define the $p$-th moment of $N_x$ by
\begin{equation}\label{eq:momentBound}
    \gamma( (1+N_x)^p) := \lim_{\lambda \to \infty} \gamma\big((1+\lambda^{-1} N_x)^{-p}(1+ N_x)^p \big) \in \R \cup\{\infty\}.
\end{equation}
This is well defined since $\{ (1+\lambda^{-1} N_x)^{-p} N_x^p \}_{\lambda \in \R_+}$ is an increasing family of elements of $\mathcal{B}$, and hence the expectation with respect to $\gamma$ is a monotone increasing function of $\lambda$.

The moment bound in \eqref{eq:momentBound} implies that the state $\gamma$ can be \textit{locally} represented by a density matrix. The precise statement is captured in Lemma~\ref{lem:loc_normal} below. States whose restriction to $\mathcal{R}^\mathrm{inv}_X$, $X\Subset \Gamma$ are represented by a density operator are called locally normal.

\begin{lm}\label{lem:loc_normal}
Let $X\Subset \Gamma$ be a finite subset and assume $\gamma$ is a state on $\mathcal{R}^\mathrm{inv}_X$ that satisfies $\gamma((1+N_x)^p)<\infty$ for some $p>0$ and all $x\in X$.
Then there exists a trace-class operator $\rho_X$ on $\mathscr{F}_X$, so that for all $A\in \mathcal{R}^\mathrm{inv}_X$ we have $\gamma(A)=\tr{\rho_X A}$. Moreover, the operator $(1+N_x)^{p/2} \rho_X^{1/2}$, $x \in X$ is Hilbert--Schmidt.   
\end{lm}

\begin{rmk}
    The fact that \eqref{eq:momentBound} implies the above Hilbert--Schmidt condition for $\rho_X$ allows us to define the expectation of finite sums of unbounded operators whose absolute value is bounded by a constant times a finite sum of $(1+N_x)^p$, $x \in \Gamma$.
\end{rmk}

\begin{proof}[Proof of Lemma~\ref{lem:loc_normal}]
Here, we provide the proof for $X=\{x\}$ for the sake of a lighter notation, as the generalization is straightforward. For fixed $\lambda >1$, $p > 0$ we define the monotone increasing and bounded function $f_{\lambda}(x) = (1+x)^p/(1+ \lambda^{-1} x)^p$, $x \geq 0$. We have
\begin{equation}
    \1_{N_x>k} \leq \frac{f_{\lambda}(N_x)}{f_{\lambda}(k)},
\end{equation}
and hence
\begin{equation}\label{eq:locallyNormal1}
\gamma(\1_{N_x>k}) \leq \lim_{\lambda \to \infty} \gamma\left(  \frac{f_{\lambda}(N_x)}{f_{\lambda}(k)} \right) \leq \frac{ \gamma( (1+N_x)^p) }{(1+k)^p}.
 \end{equation}
In particular, we have
 \begin{equation}\label{eq:number-sum}
1 = \gamma(\1) = \sum_{n=0}^k \gamma(\1_{N_x=n}) +      \gamma(\1_{N_x>k}) = \sum_{n=0}^\infty \gamma(\1_{N_x=n}),
 \end{equation}
and hence $\gamma$ has almost surely finitely many particles at $x \in X$. We now use this information to construct a density operator $\rho$ for $\gamma$. 
 
 Applications of the Cauchy-Schwarz inequality for $(A,B)\mapsto\gamma(A^*B)$ and \eqref{eq:locallyNormal1} show 
 \begin{equation}
|\gamma(\1_{N_x>k}A)| \leq (1+k)^{-p/2} \gamma((1+N_x)^p)^{1/2}\gamma(A^*A)^{1/2} .
 \end{equation}
 Any element $A\in \mathcal{R}^\mathrm{inv}_x$ commutes with $N_x$, which allows us to write
\begin{equation}\label{eq:LocallyNormal1}
\gamma(A)=\sum_{n=0}^\infty \gamma\big(\1_{N_x=n}A \big)=\sum_{n=0}^\infty \gamma\big(\1_{N_x=n}A \1_{N_x=n}\big) .  
\end{equation}
We denote $\psi_n = (1/\sqrt{n!}) (a_x^*)^{n} \Omega_x$ with the vacuum vector $\Omega_x\in \mathscr{F}_x$ and note that $ | \psi_n \rangle \langle \psi_n | \otimes \1_{\mathscr{F}_{\Gamma \backslash\{ x\} } } \cong \1_{N_x = n} \in \mathcal{R}^\mathrm{inv}_x$, where $| \psi_n \rangle \langle \psi_n |$ denotes the orthogonal projection onto $\psi_n$. The unitary equivalence in this identity is the one in \eqref{eq:ExponentialIdentity}. Using \eqref{eq:LocallyNormal1}, we write
\begin{equation}
\gamma(A)=\sum_{n=0}^\infty \rho_{nn}a_{nn} \quad \text{ with } \quad \rho_{nn} = \gamma( | \psi_n \rangle \langle \psi_n | ) \quad \text{ and } \quad a_{nn} = \langle \psi_n, A \psi_n \rangle.
\end{equation}

The numbers $\rho_{nn}$ define an operator $\rho$ on $\mathscr{F}_x$ by $\langle \psi_n, \rho \psi_m \rangle= \rho_{nn}\delta_{nm}$. This operator is diagonal in the particle number basis, with nonnegative entries, and thus a positive definite operator.
An application of \eqref{eq:number-sum} shows $\sum_{n=0}^\infty\rho_{nn}=1$. Since $\rho$ is positive definite, this implies it is trace-class with trace equal to one.

It remains to prove the Hilbert-Schmidt condition.
We have, with $f_\lambda$ as above, by monotone convergence
\begin{equation}
\gamma((1+N_x)^p)= \lim_{\lambda \to \infty} \sum_{n=0}^\infty \rho_{nn} f_\lambda(n) = \sum_{n=0}^\infty \rho_{nn} (n+1)^p= \Tr{ \rho^{1/2} (1+N_x)^p \rho^{1/2}},
\end{equation}
which is finite by assumption.
This proves the claim.
\end{proof}

\section{Existence of dynamics}
\label{sec:existenceOfDynamics}
The goal of this section is to prove the existence of dynamics on $\mathcal{B}$. 
The following proposition is the analogue of \cite[Thm~4.6]{buchholz2018}, on a lattice. Since, in contrast to the setting of particles moving in $\mathbb{R}^d$, the Laplacian is a bounded operator, the proof is slightly simpler.
We do not expect any proper subalgebra of $\mathcal{R}^\mathrm{inv} \subset \mathcal{A} \subsetneq \mathcal{B}$ to be invariant under time evolution generated by generic Hamiltonians of the form~\eqref{eq:localHamiltonian} (for a given Hamiltonian $H$ one can find such an algebra, depending on $H$, cf.~\cite[Thm. 4.6]{buchholz2018}).
The reason is that commutators with the pair potential $V$, which appear in the expansion of the dynamics, will add operators acting on more and more particles, while the free time evolution will spread these over larger and larger sets (see Example~\ref{ex:continuity} below). This leads to a loss of locality in exactly the way described in Example~\ref{ex:R(f)}. 

\begin{prop}\label{prop:dynamics}
	Let $H$ be defined as in \eqref{eq:localHamiltonian}, and assume the interaction potential $v:\Gamma \times \Gamma\to \R$ is bounded, and satisfies $\lim_{d(x,y)\to \infty} v(x,y) = 0$. For all $t\in \R$ the map
	\begin{equation*}
		\tau_t(A)=\ue^{\ui Ht}A\ue^{-\ui Ht}
		\label{eq:HeisenbergDynamics}
	\end{equation*}
	defines a $*$-automorphism on $\mathcal{B}$.
\end{prop}

\begin{proof}
 The Hamiltonian $H$ is self-adjoint and $t \mapsto \ue^{-\ui Ht}$ therefore is a one-parameter unitary group acting on $\mathscr{F}$. For any $A\in \mathcal{B} \subseteq \mathcal{L}(\mathscr{F})$ we have $\ue^{\ui Ht} A \ue^{-\ui Ht} \in \mathcal{L}(\mathscr{F})$ and it is also clear that this map has the necessary algebraic properties to qualify as a $*$-isomorphism. It therefore only remains to show that if $A \in \mathcal{B}$ then $\ue^{\ui Ht} A \ue^{-\ui Ht} \in \mathcal{B}$.
 
 In view of  Lemma~\ref{lem:structureResolventAlgebra}, we first show that $\ue^{\ui Ht} A \ue^{-\ui Ht}\vert_{\mathscr{F}_n}\in \mathfrak{K}_n$ holds for every $n \in \mathbb{N}_0$. For $A \in \mathcal{B}$ we denote $A_n = A \vert_{\mathscr{F}_n}$. The time evolution can be written as
 \begin{equation}
 	\ue^{\ui Ht} A \ue^{-\ui Ht} \mathds{1}_{N = n} = \ue^{\ui H_nt} A_n \ue^{-\ui H_nt} \mathds{1}_{N = n}
 	\label{eq:timeEvolutionAn}
 \end{equation}
 with the restriction 
 \begin{equation}
 	 H_n= \sum_{j=1}^n -\Delta_{j} + 2\sum_{1 \leq i < j \leq n} v(x_i,x_j)=: T_n + V_n,
 	\label{eq:nParticleHamiltonian}
 \end{equation}
 of $H$ in \eqref{eq:localHamiltonian} to $\mathscr{F}_n$. Here, $\Delta_j$ is the discrete Laplacian 
 \begin{equation}
 	\Delta f(x) = \sum_{ \myfrac{y \in \Gamma }{ d(x,y) = 1 }} f(y)
 	\label{eq:discreteLaplacian}
 \end{equation}
 on $\Gamma$ acting on the $j$-th coordinate $x_j \in \Gamma$.
The time evolution $t\mapsto \ue^{\ui H_nt} A_n \ue^{-\ui H_nt}$ is then the unique fixed point of the map 
\begin{align}
 F_n:C(\R, \mathcal{L}(\mathscr{F}_n))& \to C(\R, \mathcal{L}(\mathscr{F}_n)) \\
 F_n(f)&=\Big(t\mapsto  e^{\mathrm{i} T_n t} A_n e^{-\mathrm{i} T_n t} + \ui \int_0^t  e^{\mathrm{i} T_n (t-s)} [V_n,f(s)]e^{-\mathrm{i} T_n (t-s)} \d s\Big). \notag
\end{align}
We highlight that $T_n$ is bounded, and hence the unitary group $e^{-\mathrm{i}T_n t}$ is norm continuous, so the integral is defined as a norm-convergent Riemann integral. 
Since $V_n$ is a bounded operator, too, we have
\begin{equation}
 \lim_{k\to \infty} F_n^k(0) = \ue^{\ui H_nt} A_n \ue^{-\ui H_nt}
\end{equation}
in operator norm, locally uniformly in $t$ (evaluating the $k$-th iteration $F_n^k(0)$ simply produces the first $k$ terms of the Dyson series).
It is thus sufficient to prove that $F_n$ leaves $C(\R, \mathfrak{K}_n)$ invariant.

All elements of the $C^*$-algebra $\mathfrak{K}_n$ can be written as
finite linear combinations of elements of the form
\begin{equation}
 \widetilde{A} = \1_{\mathscr{F}_{n-m}} \otimes_{\mathrm{s}} K_m
 \label{eq:specialElement}
\end{equation}
with some $0 \leq m \leq n$ and a compact operator $K_m\in \mathcal{L}(\mathscr{F}_m)$.
Since $F_n$ is affine and continuous, it therefore suffices to show that $F(\widetilde{A})(t) \in \mathfrak{K}_n$ for any given $t$.
It is obvious that $ e^{\mathrm{i} T_n t} \widetilde A e^{-\mathrm{i} T_n t}= \1\otimes_{\mathrm{s}}\ue^{\ui t T_m} K_m  \ue^{-\ui t T_m}\in \mathfrak{K}_n$.
Moreover, we have (note that here the tensor product is not symmetrized)
 \begin{align}
  &\int_0^t \ue^{\ui T_n(t-s)}  [V_n, \1_{\mathscr{F}_{n-m}} \otimes K_m] \ue^{-\ui T_n(t-s)} \d s \label{eq:definitionDynamics1} \\
  &\qquad= 2\sum_{\myfrac{1 \leq i < j \leq n :}{ i,j > n-m}} \int_0^t \ue^{\ui T_n(t-s)} \Big[v(x_i,x_j), \1_{\mathscr{F}_{n-m}} \otimes K_m \Big]\ue^{-\ui T_n(t-s)} \ds \nonumber \\
  &\qquad\qquad+2 \sum_{\myfrac{1 \leq i < j \leq n :}{  i\leq n-m<j}} \int_0^t \ue^{\ui T_n(t-s)} \Big[v(x_i,x_j),\1_{\mathscr{F}_{n-m}} \otimes K_m \Big]\ue^{-\ui T_n(t-s)} \ds. \notag
   \end{align}
 Since $v$ is a bounded function, the integrand of the first term is of the form $\1_{\mathscr{F}_{n-m}}\otimes Q_m(t)$ with a compact operator $Q_m \in \mathcal{L}(\mathscr{F}_m)$. Moreover, we claim that $\lim_{n \to \infty} v(n) = 0$ implies that the integrand of the second term is of the form $\mathds{1}_{\mathscr{F}_{n-m-1}} \otimes Q_{m+1}$ (up to ordering of the variables) with a compact operator $Q_{m+1} \in \mathcal{L}(\mathscr{F}_{m+1})$. To see this, we approximate $K_m$ in norm by a Hilbert--Schmidt operator and $v(x_i,x_j)$ in norm by multiplication with a function $\tilde v(\cdot, x_j)$ with compact support. One easily checks that the relevant commutator with the approximated objects is Hilbert--Schmidt in $m+1$ coordinates and acts as the identity in the remaining coordinates. Since the ideal of compact operators is norm closed, this  proves the claim.
 Summing over the different terms of the symmetrization $\1_{\mathscr{F}_{n-m}} \otimes_{\mathrm{s}} K_m$ then proves that $F_n$ leaves the set of $\mathfrak{K}_n$-valued functions invariant.
 
 It remains to prove that $F_n$ preserves the consistency condition of Proposition~\ref{prop:Buchholz}, that is, $\kappa_n \circ F_n = F_{n-1}\circ \kappa_n$. It is clear that $\kappa_n(\ue^{\ui T_n t} A_n \ue^{-\ui T_n t})=\ue^{\ui T_{n-1} t}\kappa_n(A)\ue^{-\ui T_{n-1} t}$ since the free evolution is a product of evolutions on each tensor factor.  We may thus ignore the free evolution for the purpose of checking consistency of the action of the interaction potential.
 Consider the term corresponding to the symmetrization of the final line from~\eqref{eq:definitionDynamics1}, ignoring additionally the integral. With the unitary representation $U_\sigma$ of the permutation group $S_n$, $n \geq 2$ we find 
 \begin{align}
  &\kappa_n\bigg(\frac{1}{n!}\sum_{\sigma \in S_n}  U_\sigma \sum_{\myfrac{1 \leq i < j \leq n :}{  i\leq n-m, j >n-m}}  \Big[v(x_i,x_j),\1_{\mathscr{F}_{n-m}} \otimes_\mathrm{s} K_m \Big]U_\sigma^*\bigg)
  \label{eq:consistencyDynamics1} \\
  &\hspace{1cm}=\frac{m(n-m)}{n!}\kappa_n\bigg(\sum_{\sigma \in S_n}  U_\sigma  \Big[v(x_{n-m},x_n),\1_{\mathscr{F}_{n-m}} \otimes K_m \Big]U_\sigma^*\bigg). \notag
 \end{align}
As above, we denote $1/n!$ times the operator in the argument of $\kappa_n$ by $\mathds{1} \otimes_{\mathrm{s}} Q_{m+1}$. With this notation, the right-hand side equals
  \begin{align}
  &m(n-m) \kappa_n(\1_{\mathscr{F}_{n-m-1}}\otimes_\mathrm{s}Q_{m+1}) = m(n-m) n^{-1}(n-m-1) \1_{\mathscr{F}_{n-m-2}}\otimes_\mathrm{s} Q_{m+1} \\
  &\hspace{1cm}=\frac{m(n-m)(n-m-1)}{n!}\sum_{\sigma \in S_{n-1}}  U_\sigma  \Big[v(x_{1},x_{n-1}),\1_{\mathscr{F}_{n-m-1}} \otimes K_m \Big]U_\sigma^* \notag\\
  &\hspace{1cm}=\frac{m(n-m-1)}{(n-1)!}\sum_{\sigma \in S_{n-1}}  U_\sigma  \Big[v(x_{1},x_{n-1}),\kappa_n(\1_{\mathscr{F}_{n-m}} \otimes_\mathrm{s} K_m) \Big]U_\sigma^* \notag \\
  &\hspace{1cm}=\frac{1}{(n-1)!}\sum_{\sigma \in S_{n-1}}  U_\sigma \sum_{\myfrac{1 \leq i < j \leq n-1 }{   i\leq n-m-1< j}}  \Big[v(x_i,x_j), \kappa_n(\1_{\mathscr{F}_{n-m}}\otimes_\mathrm{s} K_m) \Big]U_\sigma^*. \notag
 \end{align}
To come to the third line, we used the definition of $\kappa_n$ in \eqref{eq:defKappaN}. The last line is exactly the term in $F_{n-1}\circ \kappa_n$ that corresponds to the term on the left-hand side of \eqref{eq:consistencyDynamics1}. The first term in~\eqref{eq:definitionDynamics1} can be treated similarly. This proves Proposition~\ref{prop:dynamics}.
\end{proof}

We highlight that Proposition~\ref{prop:dynamics} makes no claim about continuity in time. Indeed, $\tau_t$ is not even weakly continuous in $t$, as illustrated by the following example.
This example also shows that $\tau_t$ allows transport of infinitely many particles over arbitrary distances in arbitrarily short times. 
\begin{exam}\label{ex:continuity}
	We choose $\Gamma = \mathbb{Z}$ and consider the one-particle problem 
	\begin{equation}
	    \mathrm{i} \partial_t \psi_t(x) = -\Delta \psi_t(x) \quad \text{ with } \quad \Delta \psi(x) = \psi(x+1) + \psi(x-1)
	\end{equation}
	and the initial condition $\psi_0(x) = \delta_{x,0}$. We note that $\Delta e^{\mathrm{i} k \cdot x} = 2 \cos(k) e^{\mathrm{i} k \cdot x}$, $k \in [-\pi,\pi)$ and $\hat{\psi}_0(k) = 1/\sqrt{2\pi} $, and hence
	\begin{align}
		\psi_t(x) &= \frac{1}{\sqrt{2 \pi}} \int_{-\pi}^{\pi} \hat{\psi}_0(k) \exp(2\mathrm{i}t \cos(k)) e^{\mathrm{i} k \cdot x} \mathrm{d} k \label{eq:secondExample1} \\
		&= \frac{1}{\pi} \int_{0}^{\pi} \exp(2 \mathrm{i} t \cos(k)) \cos(k x) \d k = I_{|x|}(2 \mathrm{i} t)
		\nonumber
	\end{align}
	with the modified Bessel function $I_{|x|}(2 \mathrm{i} t)$, see \cite[9.6.19]{Abramowitz1972}. For fixed $x \in \mathbb{Z}$ the set of times $t$ such that $\psi_t(x) = 0$ holds is countable. For given $x$ we can therefore always pick $t>0$ as small as we wish such that $\psi_t(x) \neq 0$.

Let $\gamma$ be a weak-$*$ accumulation point of the states with $m$ particles at the origin. In particular, for a function $f:\N_0\to \R$ such that $f(n)$ admits a limit for $n\to\infty$ and $x\in \Z$
	\begin{equation}
		\gamma(f(N_x)) = \lim_{m \to \infty} \frac{1}{m!} \langle \Omega, a_0^m f(N_x) (a_0^*)^m \Omega \rangle
		\label{eq:secondExample2}
	\end{equation}
	with the vacuum vector $\Omega \in \mathscr{F}$.
	A short computation using that $H_0 \Omega = 0$ shows
	\begin{align}
		\gamma(\tau_t(f(N_x))) &= \lim_{m \to \infty} \frac{1}{m!} \big\langle \Omega, a_0^m e^{\mathrm{i}H_0 t} f(N_x) e^{-\mathrm{i}H_0 t} (a_0^*)^m \Omega \big\rangle \label{eq:secondExample3} \\
		&= \lim_{m \to \infty} \frac{1}{m!} \big\langle \Omega, (e^{-\mathrm{i}H_0 t} a_0 e^{\mathrm{i}H_0 t})^m f(N_x)  (e^{-\mathrm{i}H_0 t} a_0^* e^{\mathrm{i}H_0 t})^m \Omega \big\rangle \nonumber \\
		&= \lim_{m \to \infty} \frac{1}{m!} \big\langle \Omega, a(\psi_t)^m  f(N_x)  a^*(\psi_t)^m \Omega \big\rangle.
		\nonumber
	\end{align}
	
	Let us specify $f(x) = 1/(1+x)$ and assume that $x \in \mathbb{Z}$ and $t>0$ are chosen such that we have $|\psi_t(x)| > 0$. We also denote by $X_m$ the classical random variable associated to the operator $a_x^* a_x$ and the vector $\psi_t^{\otimes m}$ via the spectral theorem. Its expectation and variance satisfy
	\begin{align}
	    \mathbf{E}(X_m) &= \langle \psi_t^{\otimes m}, N_x \psi_t^{\otimes m} \rangle = m | \psi_t(x) |^2, \label{eq:secondExample3b} \\
	    \mathbf{Var}(X_m) &= \langle \psi_t^{\otimes m}, N_x^2 \psi_t^{\otimes m} \rangle - ( \langle \psi_t^{\otimes m}, N_x \psi_t^{\otimes m} \rangle )^2 = m ( |\psi_t(x)|^2 - |\psi_t(x)|^4 ). \notag
	\end{align}
	The right-hand side of \eqref{eq:secondExample3} with our choice of $f$ satisfies the upper bound
	\begin{align}
	    \mathbf{E}( 1/(X_m+1) ) &= \frac{1}{\mathbf{E}(X_m)} - \mathbf{E} \left( \frac{X_m+1-\mathbf{E}(X_m)}{\mathbf{E}(X_m)(1+X_m)} \right) \leq \frac{1}{\mathbf{E}(X_m)} + \mathbf{E}\Big[\frac{|\mathbf{E}(X_m)-X_m|}{ \mathbf{E}(X_m) } \nonumber \Big]\\
	    &\leq \frac{1}{\mathbf{E}(X_m)} + \frac{\sqrt{\mathbf{Var}(X_m)}}{\mathbf{E}(X_m)} = \frac{1}{m | \psi_t(x) |^2} \left( 1 + \sqrt{ m ( |\psi_t(x)|^2 - |\psi_t(x)|^4 ) } \right). \label{eq:secondExample3c}
	\end{align}
	The right-hand side goes to zero as $m \to \infty$, which allows us to conclude that for $x\neq 0$
	\begin{equation}
	    \gamma(\tau_t(1/(1+N_x))) = 0\neq 1 =	\gamma(1/(1+N_x)),
	    \label{eq:secondExample3d}
	\end{equation}
	and hence $\tau_t$ is not weakly continuous. The interpretation is that there are infinitely many particles at site $x \in \mathbb{Z}$ after an arbitrarily short time, although initially there are only particles at the origin. Hence, we cannot expect any locality properties of the dynamics without further assumptions. 
\end{exam}
\section{Propagation bounds and finite-volume approximations}\label{sect:approx}
In this section, we show that the dynamics $\tau_t(A)$ of a local observable $A \in \mathcal{R}^\mathrm{inv}_X$, with $X \Subset \Gamma$, can be approximated by a local time evolution $\tau_{t}^{\Lambda}(A) = e^{\mathrm{i}H_\Lambda t} A e^{-\mathrm{i}H_\Lambda t}$ for $X \subset \Lambda \Subset \Gamma$, where $H_{\Lambda}$ is the local Hamiltonian defined in~\eqref{eq:localHamiltonian}. This approximation holds, provided the relevant quantities are tested in a state with uniform bounds on the moments of the local particle number, and that does not have a cluster of particles ``at infinity''. The precise statement is given in Theorem~\ref{thm:unif_approx} below.

In the following, we will consider states on a finite or infinite subset $\Lambda\subseteq\Gamma$, which satisfy
\begin{equation}\label{eq:moment_bound}
 \sup_{x\in \Lambda} \gamma_\Lambda ((N_x+1)^p) \leq M
\end{equation}
for some $p\geq 1$ and $M>0$, independent of $\Lambda$.
Moreover, we require that the state can be approximated by normal states that also satisfy these bounds.

\begin{defn}\label{def:gamma_admit}
A state $\gamma$ on $\cB_\Lambda$ is $(M,p)$-admissible if there exists a net $(\gamma_\alpha)$ of normal states on $\cB_\Lambda$ converging to $\gamma$ in the weak$^*$ topology that satisfy the bound~\eqref{eq:moment_bound} for the given $M$ and $p$.
\end{defn}

\begin{rmk}
 If $\gamma$ is a state on $\cB$ satisfying the moment bound~\eqref{eq:moment_bound}, then Lemma~\ref{lem:loc_normal} implies that its restrictions to finite subsets $\Lambda \Subset \Gamma$ are normal. If we extend these to $\cB$ by
 \begin{equation}
  \gamma_\Lambda:= \gamma\vert_\Lambda \otimes \vert \Omega_{\Lambda^c} \rangle \langle \Omega_{\Lambda^c}\vert,
 \end{equation}
with the vacuum vector $\Omega_{\Lambda^c}$ of the Fock space $\mathscr{F}_{\Lambda^{\mathrm{c}}}$ (compare~\eqref{eq:ExponentialIdentity}), we obtain a net of normal states, where the index set is ordered by inclusion.
These restrictions evidently inherit the moment bound~\eqref{eq:moment_bound}, so $\gamma$ is $(M,p)$ admissible if it is the weak$^*$-limit of its local restrictions.

The point of the approximation condition is to exclude states that violate the moment condition ``at infinity'', e.g., accumulation points of states with $n$ particles at site $n\in \Z$. The particles moving to infinity are invisible to the local moment condition~\eqref{eq:moment_bound} but can lead to pathological behavior of the dynamics due to the infinite propagation speed, by the same mechanism as in Example~\ref{ex:continuity}
\end{rmk}

Our main approximation result is captured in the following theorem.

\begin{thm}\label{thm:unif_approx}
Let $v\in \ell^\infty(\Gamma\times \Gamma)$ be an interaction potential with $v(x,y)=0$ if $d(x,y)\geq r$ and assume that $p> 2d+2$, $M>0$, $R>0$ and $T>0$.
For any $\veps>0$ there exists $m_0\in \N$ so that for all $m\geq m_0$,
and all (finite or infinite) subsets $\Lambda \subseteq \Gamma$ and $X[(2m+1)r]\Subset \Lambda$ with $\mathrm{diam}(X)\leq R$,  and all $(M,p)$-admissible states $\gamma_\Lambda$ on $\cB_\Lambda$
we have
\begin{equation*}
\sup_{|t|\leq T} \left|  \gamma_\Lambda\Big(\big((\tau_t^\Lambda(A)-\tau^{X[2mr]}_t(A)\big)B\Big) \right| \leq \veps \|A\| \|B\|
\end{equation*}
for all $A\in \mathcal{R}^\mathrm{inv}_X$, $B\in \cB_\Lambda$.
\end{thm}
The proof of Theorem~\ref{thm:unif_approx} is given in Section~\ref{sect:proof-LR} below. This result is a Lieb--Robinson-type bound in the sense that it shows locality of the dynamics. 
Its proof yields an algebraic decay of time-evolved correlations as a corollary. 

\begin{cor}\label{cor:LRB}
Assume the hypothesis of Theorem~\ref{thm:unif_approx} and let $A\in \mathcal{R}^\mathrm{inv}_X$, $B\in \mathcal{R}^\mathrm{inv}_Y$ with $X,Y \Subset \Gamma$. Then
\begin{equation*}
\big|\gamma\big([\tau_t(A),B]\big)\big|\leq C(1+ d(X,Y))^{d-p/2+1} \|A\| \|B\|,
\end{equation*}
for a constant $C$ depending on $t$, $\Gamma,|X|,|Y|$, $v$ and $M,p$, but not the specific state $\gamma$.
\end{cor}
This result can be interpreted as a generalization of~\cite[Cor.~4.7]{Buchholz2014}, which allows for a larger class of states and gives a quantitative bound. Its proof will be given after that of Theorem~\ref{thm:unif_approx}.

Theorem~\ref{thm:unif_approx} implies the existence of the thermodynamic limit for the dynamics when tested with admissible states $\gamma$. Before providing the precise statement, we introduce the following notation. We say that a sequence $\{\Lambda_n\}_{n \in \mathbb{N}}$ of finite subsets $\Lambda_n \Subset \Gamma$ converges to $\Gamma$ (or exhausts $\Gamma$), denoted by $ \Lambda_n \to \Gamma$, if for every finite  $X\Subset \Gamma$ there is $n_0\in \N$ so that $X\subseteq \Lambda_n$ for all $n\geq n_0$.

\begin{cor}[Existence of thermodynamic limit]
\label{cor:existenceThermodynamicLimit}
     Assume that the interaction potential $v$ has finite range, and let $\gamma$ be a state on the Buchholz algebra $\mathcal{B}$ that is $(M,p)$-admissible with some $p > 2d + 2$ $M>0$. Let $X \Subset \Gamma$ and assume that $\{ \Lambda_n \}_{n \in \mathbb{N}}$ with $ X \subset \Lambda_n \subset \Gamma$ is a sequence of sets that converges to $\Gamma$. Then we have
\begin{equation}
    \lim_{n \to \infty} \sup_{ \substack{A \in \mathcal{R}^\mathrm{inv}_X \\ \Vert A \Vert \leq 1}} \sup_{ \substack{B \in \mathcal{R}^\mathrm{inv} \\ \Vert B \Vert \leq 1}} \sup_{|t| \leq T} \left| \gamma\Big(\big(\tau_t(A)-\tau^{\Lambda_n}_t(A)\big)B\Big) \right| = 0.
    \label{eq:existenceThermodynamicLimit}
\end{equation}
\end{cor}

Concerning Theorem~\ref{thm:unif_approx} and Corollary~\ref{cor:existenceThermodynamicLimit} we have the following remarks.

\begin{rmk}[Topology of convergence] \label{rem:topology}
    Theorem~\ref{thm:unif_approx} with $\Lambda=\Gamma$ actually implies a slightly stronger statement than that provided in Corollary~\ref{cor:existenceThermodynamicLimit}. For given $p>2d+2$, $M\geq 1$, we define a family of norms on $\mathcal{R}^\mathrm{inv}$ by
\begin{equation}
 \|A\|_{M,p}=  \sup_{\myfrac{\gamma\,\text{is}}{ (M,p)-\text{admissible}}} \sup_{\substack{B\in \mathcal{R}^\mathrm{inv} \\ \Vert B \Vert \leq 1}} |\gamma(AB)|.
\end{equation}
Since the statement of Theorem~\ref{thm:unif_approx} is also uniform in $\gamma$, we know that $\tau_t^{\Lambda_n}(A)\to \tau_t(A)$ in this family of norms, where the limit is taken over any sequence $\{ \Lambda_n \}_{n \in \mathbb{N}}$ with $\Lambda_n \to \Gamma$. In other words, the convergence in \eqref{eq:existenceThermodynamicLimit} is also uniform in $\gamma$ if one takes the supremum over all states satisfying \eqref{eq:moment_bound} with the same constants $p$ and $M$.
\end{rmk}
\begin{rmk}[Strong convergence of dynamics]
\label{rmk:strongConvergence}
If we restrict attention to normal states $\gamma(A)=\Tr{\rho A}$ that are invariant, $\ue^{-\ui Nt}\rho \ue^{\ui Nt}=\rho$, and satisfy the moment bound in~\eqref{eq:moment_bound}, Theorem~\ref{thm:unif_approx} implies
\begin{equation}
 \rho \tau_t^{\Lambda_n}(A) \to \rho \tau_t(A)
\end{equation}
in trace norm for any sequence $\Lambda_n \to \Gamma$. 
Choosing $\rho$ of rank one (and using the normalisation $H_\Lambda \psi=H\psi$ for $\psi$ supported in a sufficiently small set to fix the phase), we obtain $\ue^{\ui tH_\Lambda}\to\ue^{\ui tH}$ in the strong operator topology. However, this does not imply convergence of $\ue^{\ui H_\Lambda t} A\ue^{-\ui H_\Lambda t}$  in the operator norm or that $\mathcal{R}^\mathrm{inv}$ is $\tau_t$-invariant.
Indeed, this does not even hold for the free time evolution: 
For $f\in \ell^2_\mathrm{fin}(\Gamma)$, 
$f_t=e^{\mathrm{i} \Delta t} f$ generally does not have finite support. As Example~\ref{ex:R(f)} shows, the free evolution of  $(1+a^*(f)a(f))^{-1}$, given by $(1+a^*(f_t)a(f_t))^{-1}$, is therefore not an element of $\mathcal{R}^\mathrm{inv}$.
\end{rmk}

Given a $\tau_t$-invariant state $\gamma$ that is also admissible, we now consider the corresponding GNS representation of the invariant resolvent algebra $\mathcal{R}^\mathrm{inv}$.
Theorem~\ref{thm:unif_approx} allows us to construct a unitary time evolution on the related von Neumann algebra, obtained by the completion of this representation in the weak operator topology. Moreover, this dynamics is continuous in the strong operator topology and defines  a W*-dynamical system, which is completely determined by the dynamics of localized observables.
\begin{thm}\label{thm:gamma_dynamics}
 Let $\gamma$ be state on $\mathcal{B}$ and assume $v,\gamma$ satisfy the hypothesis of Theorem~\ref{thm:unif_approx}.
 Assume moreover that $\gamma(\tau_t(A))=\gamma(A)$ for $A\in \cB$.
 Let $(\sH_\gamma, \pi_\gamma, \Omega_{\gamma})$ be the corresponding GNS representation of $\mathcal{R}^\mathrm{inv}$. Then there exists a unique strongly continuous unitary group $U_\gamma(t)$ on $\sH_\gamma$ so that for all $A,B\in \mathcal{R}^\mathrm{inv}$
 \begin{equation*}
\langle U_{\gamma}(t)\pi_\gamma(A) \Omega_{\gamma},\pi_\gamma(B) \Omega_{\gamma}  \rangle_{\mathscr{H}_\gamma}= \gamma(\tau_t(A^*) B).
\end{equation*}
Moreover, for $X\Subset \Gamma$ and $A\in \mathcal{R}^\mathrm{inv}_X$, 
\begin{equation*}
 U_{\gamma}(t)\pi_\gamma(A) \Omega_{\gamma} = \lim_{n \to \infty} \pi_\gamma(\tau_t^{\Lambda_n}(A)) \Omega_\gamma,
\end{equation*}
in the norm of $\sH_\gamma$, where the limit is taken over any sequence $\{ \Lambda_n \}_{n=1}^{\infty}$ with $X \subseteq \Lambda_n \Subset \Gamma$ and $\Lambda_n \to \Gamma$.
\end{thm}
In combination with Example~\ref{ex:continuity} above, this shows that the only reason for the non-locality of the dynamics $\tau_t$ is an accumulation of a large number of particles at one lattice site, or at infinity.

\begin{rmk}
The above result can be interpreted in the following way. If we choose $\gamma$ as a $\tau_t$-invariant, admissible state and describing infinitely many particles, then $\mathscr{H}_\gamma$ carries a representation of $\mathcal{R}^\mathrm{inv}$ that is inequivalent to the representation on the Fock space we started with. In this case the unitary group $U_\gamma$ and the associated generator $H_\gamma$ describe the dynamics of (essentially) finitely many excitations relative to $\gamma$. This group can be expressed as the limit of localized dynamics on the algebra. The $\tau_t$-invariance of $\gamma$ is needed, because for a far-from-equilibrium state $\gamma$ one can expect that $\gamma\circ\tau_t$ and $\gamma$ yield different descriptions for an infinite number of particles. This yields inequivalent representations of $\mathcal{R}^\mathrm{inv}$ in the GNS Hilbert spaces of the state at different times and precludes the implementation of the dynamics in $\mathscr{H}_\gamma$. Note that the existence of the unitary $U_\gamma$ can already be obtained from representation theory \cite[Cor. 2.3.17]{BraRob1}, but the continuity requires the local approximation of Theorem 4.1.
\end{rmk}

Because the proof of Theorem~\ref{thm:gamma_dynamics} is rather short (admitting Theorem~\ref{thm:unif_approx}), we give it immediately.

\begin{proof}[Proof of Theorem~\ref{thm:gamma_dynamics}]
 Let $X \Subset \Gamma$ and choose a sequence of sets $\{ \Lambda_{n} \}_{n=1}^{\infty}$ with $X \subseteq \Lambda_{n} \Subset \Gamma$ for all $n \in \mathbb{N}$ and $\Lambda_n \to \Gamma$. For $A \in \mathcal{R}^\mathrm{inv}_X$ we have
 \begin{equation}
  \| \pi_\gamma(\tau_t^{\Lambda_k}(A)) \Omega_\gamma- \pi_\gamma(\tau_t^{X[2mr]}) \Omega_\gamma\|_{\mathscr{H}_\gamma}^2 = |\gamma(B(\tau_t^{\Lambda_k}(A)-\tau_t^{\Lambda_{\ell}}(A) ) )|,
 \end{equation}
 with $B=\tau_t^{\Lambda_k}(A^*)-\tau_t^{X[2mr]}(A^*)\in \cR^\mathrm{inv}_{\Lambda_k \cup X[2mr]}$.
An application of Theorem~\ref{thm:unif_approx} shows that the right-hand side is smaller than $\veps$ provided $X[(2m+1)r]\subseteq \Lambda_k$ and $m$ is sufficiently large. We conclude that $\pi_\gamma(\tau_t^{\Lambda_n}(A)) \Omega_\gamma$ is a Cauchy sequence in $\sH_\gamma$ and denote its (norm) limit by $\Psi_{A,t}$. Assume that $B \in \mathcal{R}^\mathrm{inv}_Y$ with some $Y \Subset \Gamma$. Then we have
\begin{align}
    \langle \Psi_{A,t}, \Psi_{B,t} \rangle_{\mathscr{H}_\gamma} &= \lim_{n \to \infty} \lim_{m \to \infty} \langle \pi_\gamma(\tau_t^{\Lambda_n}(A)) \Omega_\gamma, \pi_\gamma(\tau_t^{\Lambda_m}(B)) \Omega_\gamma \rangle_{\mathscr{H}_\gamma} \label{eq:vonNeumannAndi1}  \\
    &= \lim_{n \to \infty} \lim_{m \to \infty} \gamma(  \tau_t^{\Lambda_n}(A^*) \tau_t^{\Lambda_m}(B) ) = \gamma(  \tau_t(A^* B) ) \nonumber \\
    & = \gamma(A^* B)= \langle \pi(A) \Omega_{\gamma}, \pi(B) \Omega_{\gamma} \rangle_{\mathscr{H}_\gamma}, \nonumber
\end{align}
where the second equality in the second line is a consequence of Corollary~\ref{cor:existenceThermodynamicLimit}, and the third line follows from the $\tau_t$-invariance of $\gamma$. 

These considerations allow us to define a linear map $U_{\gamma}(t)$ on the dense set $D = \{ \pi_\gamma(A)\Omega_\gamma: A\in \mathcal{R}^\mathrm{inv}_X, X\Subset \Gamma\}$ by $U_\gamma(t)\pi_\gamma(A)\Omega_\gamma=\Psi_{A,t}$. From \eqref{eq:vonNeumannAndi1} we know that it is an isometry and can be extended uniquely to $\sH_\gamma$ by continuity.
By considering $\langle \Psi_{A,t}, \Psi_B\rangle$ in the same way, we see that $U^*_\gamma(t)=U_\gamma(-t)$.
Together with $\tau$-invariance of $\gamma$, this implies the identity $U_{\gamma}(t) U_{\gamma}(s) =U^*_\gamma(-t)U_\gamma(s)= U_{\gamma}(t+s)$, which implies unitarity of $U_\gamma$.
It remains to prove the strong continuity of $U_{\gamma}(t)$ in time. 

Since $U_{\gamma}(t)$ is unitary it suffices  to prove weak continuity in time. It is also sufficient to test with vectors in $D$. Assume that $A \in \mathcal{R}^\mathrm{inv}_X$, $B \in \mathcal{R}^\mathrm{inv}_Y$ with $X,Y \Subset \Gamma$. For given $\veps > 0$ we choose $n \in \mathbb{N}$ large enough such that  
 \begin{equation}
 \label{eq:vonNeumannAndi2}
 \sup_{|t|\leq 1} \big| \gamma\big(B^*(\tau_t(A)-\tau^{\Lambda_n}_t(A))\big)\big|<\veps/2
 \end{equation}
holds. From Lemma~\ref{lem:loc_normal} we know that there exists a trace-class operator $\rho_{\Lambda_n}$ with
\begin{equation}
 \gamma\big(B^* \tau^{\Lambda_n}_t(A) \big) = \Tr{\rho_{\Lambda_n} B^* \ue^{\ui t H_{\Lambda_n}} A\ue^{-\ui t H_{\Lambda_n}}}.
\end{equation}
To show the continuity of this expression with respect to $t$, we approximate $\rho_{\Lambda_n}$ in trace norm by an operator with finite rank and then use the strong continuity of $e^{\mathrm{i} t H_{\Lambda_n} }$. 
Hence,
\begin{equation}
    | \gamma\big(B^* \tau^\Lambda_t(A) \big) - \gamma\big(B^* A \big) | < \veps/2 
\end{equation}
provided $|t|$ is small enough. In combination with \eqref{eq:vonNeumannAndi2}, this shows 
\begin{equation}
 \big|\langle \pi_\gamma(B)\Omega_\gamma, (U_\gamma(t)-1) \pi_\gamma(A)\Omega_\gamma\rangle\big| < \veps 
\end{equation}
for the same values of $t$. We conclude the weak continuity of $t\mapsto U_{\gamma}(t)$ and thus the claim.
\end{proof}

\subsection{Proof of Theorem~\ref{thm:unif_approx}}\label{sect:proof-LR}

The strategy for proving Theorem~\ref{thm:unif_approx} follows ideas in~\cite{KVS2024}. Since we do not need control on the precise dependence on $t$ we can simplify the analysis considerably.
To control the dynamics of a local observable $A \in \mathcal{R}^\mathrm{inv}_X$, $X \Subset \Gamma$, we first show that the moment bound in \eqref{eq:moment_bound} can be propagated in time.
This allows us to introduce a particle number cutoff in the part of the Hamiltonian acting on degrees of freedom in the enlarged region $X[(2m+1)r]$ (see \eqref{eq:enlargement} for the notation $X[\ell]$) in a controlled way. This step in the proof relies crucially on the fact that $A$ commutes with the number operator $N_X$ (or at least changes it in a controlled way).
With this cutoff, we can use Lieb--Robinson bounds to control the propagation through the boundary region $X[2mr] \setminus X$ and prove the result. 
We start by proving the propagation bound, using a Gr{\o}nwall estimate.

\begin{prop}\label{prop:propagate_N}
Assume that $v$ is bounded and $\lim_{d(x,y) \to \infty} v(x,y) = 0$.
For all $p\geq 1$ there exists $\eta\geq 2$ so that for all $\Lambda\subset \Gamma$, and all $(M,p)$-admissible states $\gamma_\Lambda$ on $\cB_\Lambda$, the time-evolved state $\gamma_{\Lambda,t}(A):=\gamma(\tau_t^\Lambda(A))$ is $(\ue^{\eta t}M, p)$-admissible.
\end{prop}
\begin{proof}
 We give the proof in the case $\Lambda=\Gamma$, $t>0$.
 Let $(\gamma_\alpha)_\alpha$ be a net of normal states satisfying the moment bound~\eqref{eq:moment_bound}, converging to $\gamma$ as per Defintion~\ref{def:gamma_admit}, and let $\rho_\alpha$ be the corresponding trace-class operators. Clearly, $\gamma_\alpha\circ \tau_t$ converges to $\gamma\circ \tau_t$, so we only need to prove the uniform moment bounds for the time-evolved states $\gamma_\alpha\circ \tau_t$.
 We first introduce a regularisation of these states. Let $\rho_{\alpha, K}:= \1_{N \leq K} \rho_\alpha \1_{N \leq K}$, then $\rho_{\alpha, K}$ converges to $\rho_\alpha$ in trace norm as $K\to \infty$.
Denote
 \begin{equation}
  \nu_{x,K}(t)=\gamma_{\alpha, K} (\ue^{\ui t H} (N_x+1)^{p} \ue^{-\ui t H})=
  \Tr{\rho_\alpha \ue^{\ui t H}\1_{N\leq R} (N_x+1)^{p} \1_{N\leq R} \ue^{-\ui t H}},
 \end{equation}
 where the last step used that $H$ commutes with $N$.
 Hence, $\nu_{x,K}(t)$ is finite for all $x, t$. We now use a Gr{\o}nwall argument to bound its value in terms of $\nu_{x,K}(0)\leq M$, which is uniform in $K$.
Duhamel's formula gives for
$\nu_{x,K}(t)$
\begin{align}
  \nu_{x,K}(t)& =  \nu_{x,K}(0) + \int\limits_0^t  \gamma_{\alpha, K}(\ue^{\ui s H}\ui [H,(N_x+1)^p] \ue^{-\ui s H})\d s \label{eq:nu-Duhamel} \\
 &= \nu_{x,K}(0) + \int\limits_0^t  \gamma_{\alpha, K}(\ue^{\ui s H}\ui [T,(N_x+1)^p] \ue^{-\ui s H})\d s.\notag
 \end{align}
 With $ (N_x+1)^p a_x=a_x N_x^p $, we have
 \begin{align}
 \ui [T,(N_x+1)^p]&= \sum_{y:d(x,y)=1} \ui [a_x^*a_y+ a_y^* a_x, (N_x+1)^p]
 \\
 &= \sum_{y:d(x,y)=1} \Big(\ui \big(N_x^p-(N_x+1)^p\big)a_x^*a_y -\ui a_y^* a_x \big(N_x^p-(N_x+1)^p\big)\Big). \notag
\end{align}
Using the operator inequality $A^*B + B^*A\leq A^*A + B^*B$ with $A= \big((N_x+1)^p-N_x^p\big)^{1/2}  a_y^*$, $B=-\ui \big((N_x+1)^p-N_x^p\big)^{1/2}  a_x^*$, we obtain
\begin{align}
 \ui [T,(N_x+1)^p] &\leq \big((N_x+1)^p-N_x^p\big)^{1/2}  a_y a_y^*
  + a_x \big((N_x+1)^p-N_x^p\big)  a_x^* \\
  &\leq \big((N_x+1)^p-N_x^p\big) (N_y+1) +  \big((N_x+2)^p-(N_x+1)^p\big)(N_x+1). \notag
\end{align}
Since $(u+1)^p-u^p\leq p (u+1)^{p-1}$ for $u\geq 0$ and $p\geq 1$, we then find with  Young's inequality
\begin{align}
 \ui [T,(N_x+1)^p] &\leq p\big((N_x+1)^{p-1}(N_y+1) + 2^{p-1} (N_x+1)^p\big) \\
 &\leq (N_y+1)^p + (p-1+p2^{p-1})(N_x+1)^p \notag.
\end{align}
Inserting this into the identity~\eqref{eq:nu-Duhamel}, we obtain the closed inequality for $\nu_R(t)=(x\mapsto \nu_{x,K}(t))$
\begin{equation}
  \nu_{x,K}(t) \leq \nu_x(0) +  \sum_{y:d(x,y)=1} \int\limits_0^t \big((p-1+p2^{p-1})\nu_{x,K}(s) + \nu_{y,R}(s) \big)\d s,
\end{equation}
and thus
\begin{equation}
  \| \nu_R(t)\|_\infty \leq  \| \nu_R(0)\|_\infty +
  p\sigma (2^{p-1}+1) \int_0^t \| \nu_R(s)\|_\infty \d s.
\end{equation}
Gr{\o}nwall's inequality then implies that
\begin{equation}
 \| \nu_R(t)\|_\infty \leq M \ue^{\eta t}
\end{equation}
with $\eta = p\sigma (2^{p-1}+1)$. Let $M_{x,\lambda}=(1+N_x)^p/(1+\lambda^{-1} N_x )^p$ be the regularised $p$-th power of $1+N_x$. Then, since $M_{x,\lambda}\leq (1+N_x)^p$,
\begin{equation}
 \gamma_\alpha(\tau_t(M_{x,\lambda}))= \lim_{K\to \infty} \gamma_{\alpha,K}(\tau_t(M_{x,\lambda})) \leq \lim_{K\to \infty} \gamma_{\alpha,K}(\tau_t(1+N_x)^p)\leq M \ue^{\eta t}.
\end{equation}
Since the right hand side is uniform in $\lambda$, this proves the claim.
\end{proof}

With this result at hand, we can now compare the dynamics $\ue^{-\ui H_\Lambda t}$ to those, where all interaction terms in the Hamiltonian in some finite region $Y$, containing the support $X$ of our observable, have been multiplied with a projection that restricts the local particle number. 
For $Y\Subset \Lambda \subset \Gamma$ and $\lambda\geq 1$ we define
\begin{equation}\label{eq:Pdef}
 P_{Y,\lambda}:= \prod_{x\in Y} \1_{N_x\leq \lambda},
\end{equation}
$P_{Y,\lambda}^\perp=\1-P_{Y,\lambda}$, and
 \begin{equation}
  \widetilde{\tau}^{\Lambda}_t(A):=\ue^{\ui t P_{Y,\lambda}H_\Lambda P_{Y,\lambda}} A\ue^{-\ui t P_{Y,\lambda}H_\Lambda P_{Y,\lambda}}.
 \end{equation}
The following lemma quantifies the influence of the cutoff $P_{Y,\lambda}$ on the dynamics.

\begin{lm}\label{lem:cutoff}
Assume that $v$ is bounded and $\lim_{d(x,y) \to \infty} v(x,y) = 0$.
There exists a constant $C_\Gamma$ so that for all $\emptyset \neq X\subseteq Y\Subset \Lambda\subseteq \Gamma$, $Z\subseteq \Lambda$, $p\geq 2$, $M\geq 1$, $(M,p)$-admissible states $\gamma_\Lambda$, $t\in \R$, and $A\in \mathcal{R}^\mathrm{inv}_X$, $B\in \cB_\Lambda$
  we have
\begin{equation*}
|\gamma_\Lambda(\tau^Z_t(A)B) - \gamma_{\Lambda}(\widetilde\tau^Z_t(P_{Y,\lambda}AP_{Y,\lambda})B)|\leq C_\Gamma
\lambda^{-p/2+1}  |Y| |X|^{p/2} \ue^{\eta |t|} \sqrt{M}  \|A\|\|B\|,
\end{equation*}
with $\eta\geq 2$ given by Proposition~\ref{prop:propagate_N}.
\end{lm}
\begin{proof}
In the following we abbreviate $P=P_{Y,\lambda}$, $\gamma = \gamma_{\Lambda}$ and choose $Z=\Lambda=\Gamma$, $t\geq 0$ for the sake of simplicity. We will separate the difference in two terms. We start by computing
\begin{align}
 &\Big|\gamma\Big( \ue^{\ui t PHP} PAP \ue^{-\ui t PHP}B-\ue^{\ui t H }A P\ue^{-\ui t PHP}B\Big)\Big| \label{eq:LRBoundAndi1} \\
 &\leq  \Big|\gamma\Big(\ue^{\ui t H }P^\perp AP\ue^{-\ui t PHP }B \Big) \Big|
 + \Big|\gamma\Big( \big(\ue^{\ui t H }-\ue^{\ui t PHP}\big)PAP\ue^{-\ui tPHP}B \Big) \Big|. \notag
\end{align}
The Cauchy-Schwarz inequality for the sesquilinear form $(A,B)\mapsto \gamma(A^*B)$ gives
\begin{align}
\big|\gamma\big(\ue^{\ui t H }P^\perp  AP \ue^{-\ui t P HP  }B\big)  \big|
\leq \sqrt{\gamma(\ue^{\ui t H } P^\perp\ue^{-\ui t H })}\|A\|\|B\| 
\end{align}
The bound $P^\perp \leq \sum_{x\in Y} \1_{N_x>\lambda}$, and  Proposition~\ref{prop:propagate_N}, imply
\begin{equation}
\sqrt{\gamma(\ue^{\ui t H } P^\perp\ue^{-\ui t H })}
\leq \lambda^{-p/2} \Big(\sum_{x\in Y}\gamma\big(\ue^{\ui t H }  N_x^p\ue^{-\ui t H }\big)\Big)^{1/2} \leq \ue^{\eta t/2}\sqrt{\lambda^{-p}M |Y|},
\end{equation}
whence
\begin{equation}\label{eq:cutoff_bound1}
 \big|\gamma\big(\ue^{\ui t H }P^\perp  AP \ue^{-\ui t P HP  }B\big)  \big|
 \leq \ue^{\eta t/2}\sqrt{\lambda^{-p}M |Y|}\|A\|\|B\|.
\end{equation}
By Duhamel's formula (whose use is justified by the same truncation argument as in Proposition~\ref{prop:propagate_N}) and $PV P^\perp=0$ (since $V_{xy}$ commutes with $N_x$), we have
\begin{align}
\label{eq:cutoff_bound_Duhamel}
 &\Big|\gamma\Big(\big(\ue^{\ui t H }-\ue^{\ui t PHP}\big)P AP\ue^{-\ui tPHP }B \Big) \Big| \\
 &\leq
 \int_0^t \Big|\gamma\Big(\ue^{\ui (t-s)H}(HP-PHP)\ue^{\ui s P HP}AP\ue^{-\ui tPHP}B \Big)\Big| \d s \notag \\
 &\leq 2\int_0^t \Big|\sum_{\myfrac{x\in Y, y\in Y[1]}{d(x,y)=1}}\gamma\big(\ue^{\ui (t-s)H}  P^\perp_{ \{x,y\}\cap Y,\lambda} T_{xy}P\ue^{\ui s P HP} AP \ue^{-\ui t PHP} B \big)\Big|  \d s \notag \\
 &\leq 2\int_0^t \sum_{\myfrac{x\in Y,y\in Y[1]}{d(x,y)=1}}\gamma\big(\ue^{\ui (t-s)H}  P^\perp_{ \{x,y\}\cap Y,\lambda} T_{xy}P T_{xy} P^\perp_{ \{x,y\}\cap Y,\lambda} \ue^{-\ui (t-s)H}\big)^{1/2} \|A\| \|B\| \d s, \notag
\end{align}
where we used that $P_{Y\setminus\{x,y\},\lambda}$  commutes with $T_{xy}$, and $P^\perp P_{Y\setminus\{x,y\},\lambda}
=P_{\{x,y\}\cap Y,\lambda}^\perp P_{Y\setminus\{x,y\},\lambda}$.
Using the bounds
\begin{equation}\label{eq:PT-bound}
\|P_{\{x,y\},\lambda} T_{xy}\|\leq 2\sqrt{\lambda(\lambda+1)}\leq 2\sqrt{2}\lambda,
\qquad
\| P_{x,\lambda} T_{xy} (N_y+1)^{-1/2}\| \leq2\sqrt{\lambda+1},
\end{equation}
and Proposition~\ref{prop:propagate_N}, we obtain for $x,y\in Y$
\begin{align}
 \gamma\big(\ue^{\ui (t-s)H}  P^\perp_{ \{x,y\},\lambda} T_{xy}P T_{xy} P^\perp_{ \{x,y\},\lambda} \ue^{-\ui (t-s)H}\big)
 \leq 8 \lambda^2 \gamma\big(\ue^{\ui (t-s)H} P^\perp_{ \{x,y\},\lambda}  \ue^{-\ui (t-s)H}\big) \leq 16 \lambda^{-p+2} M \ue^{(t-s)\eta},
\end{align}
and for $y\notin Y$,
\begin{align}
 \gamma\big(\ue^{\ui (t-s)H}  P^\perp_{ \{x\},\lambda} T_{xy}P T_{xy} P^\perp_{ \{x\},\lambda} \ue^{-\ui (t-s)H}\big) &
 \leq 4 (\lambda+1)  \gamma\big(\ue^{\ui (t-s)H} P^\perp_{ \{x\},\lambda}(N_y+1)  \ue^{-\ui (t-s)H}\big) \\
 &\leq 4(\lambda+1)\lambda^{-p+1 }M \ue^{(t-s)\eta}.\notag
\end{align}
Bounding the number of terms in the sum by $|Y| \sigma$ and integrating in $s$ then shows that
\begin{align}
 &\Big|\gamma\Big(\big(\ue^{\ui t H }-\ue^{\ui t PHP}\big)P A P\ue^{-\ui tPHP}B \Big) \Big|
 \leq 8  |Y|  \lambda^{-p/2+1}  (\ue^{\eta t/2 }-1)  \sigma  \sqrt{M}\|A\| \|B \|.
 \label{eq:cutoff_bound2}
\end{align}

We now come to the term where the difference of the evolutions is between $A$ and $B$, i.e.,
\begin{equation}
  \big|\gamma\big( \ue^{\ui t H } AP \ue^{-\ui t PHP}B-\ue^{\ui t H }A \ue^{-\ui t H}B\big)\big|.
\end{equation}
As above, we have
\begin{align}
 \big|\gamma\big(\ue^{\ui t H }  A \ue^{-\ui t H } P^\perp  B\big)  \big|
 \leq \sqrt{\gamma\big(\ue^{\ui t H } A\ue^{-\ui t H } P^\perp\ue^{\ui t H } A^* \ue^{-\ui t H }\big)}  \|B\|. 
\end{align}
Consider the positive linear functional $\gamma_A$ on $\mathcal{R}^\mathrm{inv}$ defined by
\begin{equation}
 \gamma_A(K)= \gamma\big(\ue^{\ui t H } A K A^* \ue^{-\ui t H }\big).
\end{equation}
For $x\notin X$ an application of Proposition~\ref{prop:propagate_N} shows
\begin{equation}
    \gamma_A(N_x^p)=  \gamma\big(\ue^{\ui t H }  N_x^{p/2} AA^*N_x^{p/2} \ue^{-\ui t H }\big) 
    \leq \|A\|^2 M \ue^{\eta t},
\end{equation}
and for $x\in X$ we have 
\begin{multline}
    \gamma_A((1+N_x)^p)\leq   \gamma\big(\ue^{\ui t H } A (1+N_X^p) A^* \ue^{-\ui t H }\big)
   \\ \leq \|A\|^2 \gamma\big(\ue^{\ui t H } (1+ N_X)^p  \ue^{-\ui t H }\big) \leq  \|A\|^2 |X|^p  M \ue^{\eta t}.
\end{multline}
To obtain the result we used $N_X^p\leq |X|^{p-1} \sum_{x\in X} N_x^p$. By applying Proposition~\ref{prop:propagate_N} to the normalised version of $\gamma_A$, we obtain
\begin{equation}
 \sup_{x\in \Gamma} \gamma_A(\ue^{\ui t H }  (N_x+1)^p  \ue^{-\ui t H }) \leq \|A\|^2 |X|^p  M \ue^{2\eta t},
\end{equation}
and thus
\begin{equation}
  \big|\gamma\big(\ue^{\ui t H }  A \ue^{-\ui t H } P^\perp B\big)  \big| \leq \ue^{ \eta t}\sqrt{\lambda^{-p}|X|^{p} M |Y|}\|A\|\|B\|.
  \label{eq:cutoff_bound3}
\end{equation}
Then, using Duhamel's formula as in~\eqref{eq:cutoff_bound_Duhamel}, we find
\begin{align}
 \Big|\gamma\Big( &\ue^{\ui t H } A\big(\ue^{-\ui t H }-\ue^{-\ui t PHP}\big)  P B \Big) \Big|
 \leq 2\int_0^t \Big|\sum_{\myfrac{x\in Y, y\in Y[1]}{d(x,y)=1}}\gamma\big(\ue^{\ui t H } A \ue^{-\ui s H} P^\perp  T_{xy}  P  \ue^{-\ui (t-s)  P HP } B \big)\Big|  \d s\notag \\
  &\qquad\leq 2 \int_0^t \sum_{\myfrac{x\in Y, y\in Y[1]}{d(x,y)=1}}\gamma_A\big(\ue^{-\ui s H} P^\perp_{\{x,y\}\cap Y,\lambda}  T_{xy}  P  T_{xy} P^\perp_{\{x,y\}\cap Y,\lambda} \ue^{\ui s H}\big)^{1/2}\|B\| \d s .
\end{align}
Using~\eqref{eq:PT-bound} together with the bounds on $\gamma_A$ derived above, we conclude as for~\eqref{eq:cutoff_bound2} that
\begin{equation}\label{eq:cutoff_bound4}
\Big|\gamma\Big( \ue^{\ui t H } A\big(\ue^{-\ui t H }-\ue^{-\ui t PHP}\big)  P B \Big) \Big|
\leq 16  |Y|  \lambda^{-p/2+1}  (\ue^{\eta t }-1)  \sigma  \sqrt{M|X|^p}\|A\| \|B \|.
\end{equation}
Combining~\eqref{eq:cutoff_bound1},~\eqref{eq:cutoff_bound2},~\eqref{eq:cutoff_bound3}, and~\eqref{eq:cutoff_bound4} then implies the claim.
\end{proof}

\begin{lm}\label{lem:local_approx_V}
Assume that $v\in \ell^\infty(\Gamma\times \Gamma)$ with $v(x,y) = 0$ if $d(x,y) \geq r \geq 1$. There exists $\kappa>0$ only depending on $r$ and $\sigma$, and for all $R>0$ there is a constant $C>0$ so that for all $X\Subset \Lambda$ with $\mathrm{diam}(X)\leq R$, $A\in \mathcal{L}(\mathscr{F}_X)$,
$t>0$, $\lambda \geq 1$ and $m\in \N$ with $ X[(2m+1)r]\subseteq \Lambda$, we have
\begin{equation*}
 \Big\|\ue^{\ui t P H_{X[2mr]} P} PAP \ue^{-\ui t P H_{X[2mr]} P} - \ue^{\ui t P H_\Lambda P} PAP \ue^{-\ui t P H_\Lambda P}\Big\|\leq
 C   \lambda m^{d}\|A\| \frac{(\kappa \lambda t)^{m+1}}{(m+1)!}
 \end{equation*}
 with $P=P_{X[(2m+1)r],\lambda}$ in~\eqref{eq:Pdef}.
\end{lm}

\begin{proof}
We start by noting that, since $H_{\Lambda\setminus X[2mr]}$ commutes with $A$ and $H_{X[2mr]}$,
 \begin{equation}
  \ue^{\ui t (PH_{\Lambda\setminus X[2mr]}P + PH_{X[2mr]}P)  }PAP \ue^{-\ui t ( PH_{\Lambda\setminus X[2mr]}P + PH_{X[2mr]}P)} = \ue^{\ui t  PH_{X[2mr]}P  }PAP \ue^{-\ui t  PH_{X[2mr]}P}.
 \end{equation}
Hence, it is sufficient to compare the evolution on the left to $\ue^{-\ui t P H_{\Lambda} P}$. The difference of the generators of the two dynamics is bounded and supported in the set $X[(2m+1)r]\setminus X[(2m-1)r]$.
It equals
\begin{equation}
 P H_\Lambda P - PH_{\Lambda\setminus X[2mr]}P - PH_{X[2mr]}P
 =  \sum_{\myfrac{ \{x,y\}\cap \Lambda\setminus X[2mr]\neq \emptyset }{
 \{x,y\}\cap  X[2mr]\neq \emptyset }} \underbrace{PT_{xy}P + PV_{xy}P}_{=:D_{xy}}=:D,
\end{equation}
and thus
\begin{multline}\label{eq:difference_LRB}
 \ue^{\ui t P H_{X[2mr]} P} PAP \ue^{-\ui t P H_{X[2mr]} P} - \ue^{\ui t P H_\Lambda P} PAP \ue^{-\ui t P H_\Lambda P} \\
 =-\ui\int_0^t  \ue^{\ui (t-s) P H_{\Lambda} P}\ue^{\ui s P H_{X[2mr]} P}\Big[ \ue^{-\ui s P H_{X[2mr]} P} D\ue^{\ui s P H_{X[2mr]} P} ,   PAP\Big]\ue^{-\ui s P H_{X[2mr]} P}\ue^{-\ui (t-s) P H_{\Lambda} P} \d s. 
\end{multline}
Since the distance of  the points $x,y$ contributing to $D$ and the support of $A$ is larger than $(2m-1)r$ and $PH_{X[2mr]}P$ is (the projection of) a sum of bounded, local terms, we can use Lieb--Robinson bounds to estimate the norm of the time-dependent commutator in the above equation.
Such a bound for quantum spin systems can be found in~\cite[Thm.~3.1]{NSY2019}. The proof of a similar bound in our setting goes along the same lines. For the sake of completeness, we provide the details.

We can remove the potential terms, which commute with the number operators, from the Hamiltonian, at the price of making the hopping time-dependent (cf.~\cite{NSY2019, KVS2024}).
This yields an improved estimate, where we can bound interaction terms by $\lambda$ instead of $ \lambda^2$. To do this, we introduce the dynamics
\begin{equation}
 U(t) = \ue^{\ui t PV_{X[2mr]}P}\ue^{-\ui t P H_{X[2mr]} P}= 1 - \ui \int_{0}^{t} \ue^{\ui s  PV_{X[2mr]} P} P T_{X[2mr]}P \ue^{-\ui s  PV_{X[2mr]} P} U(s) \d s
 \end{equation}
 with time-dependent generator
 \begin{equation}
\widetilde {T}_{X[2mr]}(s)=\ue^{\ui s  PV_{X[2mr]} P} P T_{X[2mr]}P \ue^{-\ui s  PV_{X[2mr]} P},
 \end{equation}
 and set $U(s,t)=U(t)U(s)^*$.
 The norm of the commutator involving $A,D$ can then be written as
 \begin{equation}
  \Big\|\Big[ \ue^{-\ui s P H_{X[2mr]} P} PDP\ue^{\ui s P H_{X[2mr]} P} ,   PAP\Big] \Big\| = \Big\|\Big[ U(s) DU(s)^* ,  \ue^{\ui s PV_{X[2mr]}P} PAP\ue^{-\ui s PV_{X[2mr]}P}\Big] \Big\|.
 \end{equation}
Since $V$ equals a sum of mutually commuting terms, we have
\begin{equation}
 \ue^{\ui s P V_{X[2mr]} P} PAP \ue^{-\ui s P V_{X[2mr]} P} = P\ue^{\ui s P V_{X[r]} P}A \ue^{-\ui s PV_{X[r]} P}P=:\widetilde {A}(s),
\end{equation}
which is the projection of an operator supported in $X[r]$. Since for any pair of neighboring points $x,y$ the potential terms involving points outside of $x[r]$ commute with $T_{xy}$ and $P$, we also have
\begin{equation} \ue^{\ui s V_{X[2mr]}}P T_{xy} P \ue^{-\ui s V_{X[2mr]}}=P \ue^{\ui s V_{x[r]}}T_{xy} \ue^{-\ui s V_{x[r]}}P=:
\widetilde{T}_{xy}(s).
\end{equation}
Note that $\widetilde{T}_{xy}(s)$ is the projection of an operator supported in $x[r]$, and in particular $\widetilde{T}_{xy}(s)$ commutes with $\widetilde A(s)$ if $x[r]\cap X[r]=\varnothing$.

For a bounded operator $B$ we denote $B(s,t)=U(s,t)  B U(s,t)^*$. Assume that $B$ commutes with $\widetilde {A}(t)$ and that there is some $x_0\in \Gamma$ so that $B$ also commutes with $\widetilde T_{xy}(s)$, $0\leq s \leq t$, unless $x\in x_0[2r]$ (as is the case for $B=D_{x_0y}$).
Then, using the Jacobi identity,
\begin{align}
 \frac{\d}{\d s} [B(s,t),\widetilde {A}(t)] =& \Big[ U(s,t)\ui\big[\widetilde T_{X[2mr]}(s), B\big]U(s,t)^*,\widetilde {A}(t)\Big] \\
 =& \sum_{\myfrac{x\in x_0[2r]}{ y\in x[1] }}\ui \Big[ U(s,t)\big[\widetilde T_{xy}(s), B\big]U(s,t)^*,\widetilde {A}(t)\Big] \notag \\
 =&\begin{aligned}[t]
 \sum_{\myfrac{x\in x_0[2r]}{ y \in x[1] }} \Big(& \Big[ \ui U(s,t)\widetilde T_{xy}(s)U(s,t)^*, \big[B(s,t), \widetilde {A}(t)\big]\Big] \notag \\
 & +\ui \Big[ B(s,t),  \big[\widetilde {A}(t), U(s,t)\widetilde T_{xy}(s)U(s,t)^*\big]\Big] \Big).
 \end{aligned}
  \notag
\end{align}
The above equation is of the form $\partial_s f(s,t) = \mathrm{i} [ K(s,t), f(s,t) ] + \mathrm{i} C(s,t)$ with a self-adjoint operator $K(s,t)$ and the final condition $f(t,t) = 0$. Its solution satisfies the bound $\Vert f(s,t) \Vert \leq \int_s^t \Vert C(s',t) \Vert \mathrm{d} s' $, and hence
\begin{equation}\label{eq:LR-intbound}
 \| [B(s,t),\widetilde {A}(t)] \| \leq  2\|B\| \int_s^t \sum_{\myfrac{x\in x_0[2r], }{ y \in x[1] }}
 \big\| \big[U(s',t)\widetilde T_{xy}(s')U(s',t)^*, \widetilde {A}(t)\big]\big\|\d s'.
\end{equation}
If the distance of $X[r]$ and $x_1\in x_0[2r]$ is larger than $r$, then $\widetilde T_{x_1 y}(s')$ with $x_1\in x_0[2r]$ is again an operator with the same properties as $B$. That is, it commutes with $\widetilde A$ and all hopping terms $\widetilde T_{xy}(t)$, where $x\notin x_1[2r]$. The above holds if $d(x_0,X) \geq 2kr$ with $k \geq 2$. Accordingly, we can iterate the inequality~\eqref{eq:LR-intbound} another $k-1$ times.
Using the trivial bound $\|[U(s,t)\widetilde T_{xy}(s)U(s,t)^*,\widetilde A(t)] \|\leq 2 \|\widetilde T_{xy}(s)\| \| A\|$ in the last iteration, this yields
\begin{multline}
 \|[B(0,t),\widetilde {A}(t)]\|
 \leq 2^{k+1}\|B\| \|  A\| \hspace{-12pt}
 \int\limits_{0\leq s_1  \dots \leq s_{k}\leq  t}
 \sum_{\myfrac{x_1\in x_0[2r]} { y_1 \in x_1[1] } } \|\widetilde T_{x_1 y_1}(s_{1})\| \cdots\hspace{-16pt} \sum_{\myfrac{x_{k}\in x_{k-1}[2r]}{y_k\in x_k[1]}} \|\widetilde T_{x_{k} y_k}(s_{k})\|  \d s_{k} \dots \d s_1 . 
\end{multline}
We know from~\eqref{eq:PT-bound} that $\| \widetilde T_{x y}(t)\|\leq 2\sqrt{2}\lambda$, and, using Definition~\ref{def:dDimGraph}, we check that the number of terms in each sum is bounded by $\sigma(\sigma(2r)^d+1)$. Thus, with $\kappa= 4\sqrt{2}\sigma(\sigma(2r)^d+1)$, we have
\begin{equation}
 \| [B(0,t),\widetilde {A}]\|\leq 2\|B\| \| \widetilde A\| \frac{(\kappa \lambda t)^{k}}{k!}.
\end{equation}

We now apply this to each summand $D_{xy}$ in~\eqref{eq:difference_LRB}, where $k=m$.
The number of such terms in $D$ is bounded by $ (1+\sigma r^d) |X[2mr]| \leq  (1+\sigma r^d) (1+\sigma (2mr)^d) |X|$.
With $\|D_{xy}\|\leq (2+\|v\|_\infty)\lambda^2$, we thus obtain
\begin{align}
\|\eqref{eq:difference_LRB}\|\leq \sum_{x,y} \|D_{xy}\|\|A\| \int_0^t \frac{(\kappa \lambda s)^{m}}{(m)!} \d s \leq \|A\|  C m^{d} \lambda  \frac{(\kappa \lambda t)^{m+1}}{(m+1)!},
\end{align}
which proves the claim.
\end{proof}

\begin{proof}[Proof of Theorem~\ref{thm:unif_approx}]
Applying first Lemma~\ref{lem:cutoff} with $Y=X[(2m+1)r]$, $Z=\Lambda, X[2mr]$, and $\lambda= m /(\ue^{2}\kappa T)$ and then Lemma~\ref{lem:local_approx_V} together with $(m+1)!\geq  m^{m+1}\ue^{-m}$, we obtain
\begin{align}
 \Big| \gamma\Big((\tau_t^\Lambda(A)-\tau^{X[2mr]}_t(A))B\Big) \Big|
 &\stackrel{\ref{lem:cutoff}}\leq \Big| \gamma\Big((\widetilde{\tau}_t^\Lambda(PAP)-\widetilde{\tau}^{X[2mr]}_t(PAP))B\Big) \Big| \label{eq:dynamics-full-approx} \\
 &\qquad + 2\lambda^{-p/2+1}  |X[(2m+1)r]| |X|^{p/2}\ue^{\eta t} \sqrt{M} C_\Gamma \|A\|\|B\| \notag \\
 & \stackrel{\ref{lem:local_approx_V}}\leq C
 \Big( m^{d+1} \ue^{-m} + m^{d-p/2+1} \Big)\|A\|\|B\|. \notag
\end{align}
Since $p/2>d+1$, taking $m$ sufficiently large makes the parenthesis smaller than $\veps$, which proves the claim.
\end{proof}

\begin{proof}[Proof of Corollary~\ref{cor:LRB}]
 This follows from~\eqref{eq:dynamics-full-approx} with the choice $m=\lfloor d(X,Y)/(2r)\rfloor-1$ and the identity
 \begin{equation}
      \gamma\big([\tau_t^{X[2mr]}(A),B]\big)=0.
 \end{equation}
\end{proof}

\section{Equilibrium states}\label{sect:KMS}

In this section we show the existence of KMS states. In particular, we show that the limit of any convergent net (or sequence) of suitable finite-volume Gibbs states at temperature $\beta > 0$ satisfies the KMS condition. Our result applies to the standard Bose--Hubbard model, see Remark~\ref{exm:Bose--Hubbard} below.

Our definition of KMS states requires moment bounds on the local particle numbers in order to be able to define the quasi-local dynamics in Theorem~\ref{thm:gamma_dynamics} and formulate the KMS condition for this dynamics. 
\begin{defn}\label{def:KMS}
 Let $\gamma$ be a state on $\mathcal{B}$ with $\gamma(\tau_t(A))=\gamma(A)$ for $A\in \mathcal{R}^\mathrm{inv}$.
 We call $\gamma$ a tempered KMS state at inverse temperature $\beta>0$ if it satisfies the moment bound~\eqref{eq:moment_bound} for some $p>2d+2$ and if
for all $A,B\in \mathcal{R}^\mathrm{inv}$  there exists a bounded, continuous function
 \begin{equation*}
  F:S_\beta:=\{z \in \C : -\beta \leq \Im z \leq 0\} \to \C,
 \end{equation*}
which is holomorphic in $\mathrm{int} (S_\beta)$, and satisfies
\begin{equation*}
 F(t)=\gamma(\tau_t(A)B), \qquad F(t-\ui \beta)=\gamma(B\tau_t(A)).
\end{equation*}
\end{defn}

 Note that we do not require $t \mapsto \tau_t(A)$ to be continuous for all $A \in \mathcal{B}$, as is required by the definition of KMS states for $C^*$-dynamical systems.
 Our form of the KMS condition should rather be compared to the usual condition (cf.~\cite{DJP2003}) for the $W^*$-dynamical system obtained from Theorem~\ref{thm:gamma_dynamics}.

The main result of this section is captured in the following theorem. 

\begin{thm}\label{thm:KMS}
Let $v$ be a bounded interaction potential with $v(x,y)=0$ if $d(x,y)\geq r$ and assume that $v$, $\mu_\Lambda \in \mathbb{R}$, and $\beta > 0$ are such that $Z_\Lambda:=\Tr{ \ue^{-\beta (H_\Lambda- \mu_\Lambda N_\Lambda)}} < + \infty$ holds for any $\Lambda \Subset \Gamma$. Assume in addition that the state
 \begin{equation}
 \gamma_\Lambda(A)= Z_\Lambda^{-1}
 \Tr{ A \ue^{-\beta (H_\Lambda- \mu_\Lambda N_\Lambda)} \otimes | \Omega_{\Lambda^{\mathrm{c}}} \rangle \langle \Omega_{\Lambda^{\mathrm{c}}} | } ,
 \label{eq:finiteVolumeGibbsStateTimesVacuum}
 \end{equation}
 where $\Omega_{\Lambda^{\mathrm{c} }}$ denotes the vacuum vector of the Fock space $\mathscr{F}_{\Lambda^{\mathrm{c}}}$ (compare~\eqref{eq:ExponentialIdentity}), satisfies
\begin{equation}
 \sup_{\Lambda \Subset \Gamma} \sup_{x\in \Lambda} \gamma_\Lambda \big((1+N_x)^p\big) <\infty
 \label{eq:MomentBoundAssumption}
\end{equation}
with some $p>2d+2$.
Then any accumulation point of $\{ \gamma_{\Lambda} \}_{\Lambda \Subset \Gamma}$ is a tempered KMS state at inverse temperature $\beta > 0$. 
\end{thm}

Since $\{ \gamma_{\Lambda} \}_{\Lambda \Subset \Gamma}$ has at least one accumulation point by the Banach--Alaoglu Theorem, Theorem~\ref{thm:KMS} directly implies the existence of KMS states.
\begin{rmk}\label{exm:Bose--Hubbard}
We expect that the hypotheses of Theorem~\ref{thm:KMS}, particularly \eqref{eq:MomentBoundAssumption}, are satisfied under the assumptions that the interaction is superstable (which is the generic case and holds if $U>0$ in~\eqref{eq:BoseHubbardHamiltonian}) and that $\lim_{n \to \infty} \mu_{\Lambda_n} = \mu \in \mathbb{R}$.

Indeed, \eqref{eq:MomentBoundAssumption} was shown to hold recently~\cite[Thm.~5.25]{sekine2026} under the hypothesis $v\geq 0$ and $\mu<-2d$ using path integral techniques (see~\cite[Theorem~6.3.17]{BraRob1} for an anologous approach to bosons in $\R^d$).
For the specific case $v(x)=U\delta_{x,0}$, \eqref{eq:MomentBoundAssumption} has also been proved for the first moment in~\cite{lemm2023thermal} and for general $p$ at high temperatures~\cite{tong2024boson}.
Bounds on the moments of the local particle number with more general hypothesis are available only for bosons in the continuum $\mathbb{R}^d$~\cite[Prop.~2.2.1]{Park1985}. In this case, the analogue of \eqref{eq:MomentBoundAssumption} holds for every $p$ under the assumption of superstability mentioned above. Like~\cite{sekine2026}, the argument uses the path integral representaion, and we expect that it can be adapted to the lattice setting.
\end{rmk}

Before we give the proof of Theorem~\ref{thm:KMS}, we state and prove two lemmas to not interrupt the main line of the argument later. In the first lemma we show that the state on $\mathcal{R}^\mathrm{inv}_{\Lambda}$ given by the density matrix in \eqref{eq:finiteVolumeGibbsStateTimesVacuum} is a tempered KMS state. 
 
 \begin{lm}\label{lem:Lambda-KMS}
 Let $\Lambda \Subset \Gamma$ and assume that $v$, $\mu_{\Lambda}$, and $\beta$ satisfy the assumptions of Theorem~\ref{thm:KMS}. Then the state given by the density matrix
 \begin{equation*}
 \rho_{\Lambda} = Z_\Lambda^{-1}
 \ue^{-\beta (H_\Lambda- \mu_\Lambda N_\Lambda)}
 \end{equation*}
 is a $(\tau_\Lambda,\beta)$--KMS state.
 \end{lm}
 \begin{proof}
 We start by noting that the state given by $\rho_{\Lambda}$ is $\tau_t^{\Lambda}$ invariant. This follows from the cyclicity of the trace, see e.g. \cite[Theorem~3.1]{Simon2005}. 
 
 For given $A, B \in \mathcal{R}^\mathrm{inv}_{\Lambda}$ and $z \in S_{\beta}$ we consider the function
 \begin{equation}
     F(z) = Z_\Lambda^{-1} \Tr{ \ue^{-\beta (H_\Lambda- \mu_\Lambda N_\Lambda)} \ue^{\mathrm{i} H_{\Lambda} z} A \ue^{-\mathrm{i} H_{\Lambda} z} B }.
     \label{eq:finiteVolumeKMS0}
 \end{equation}
 It is well defined because for $z = t - \mathrm{i} s$ with $t \in \mathbb{R}$ and $0 \leq s \leq \beta$ an application of Hölder's inequality for traces, see e.g. \cite[Theorem~2.8]{Simon2005}, shows
 \begin{equation}
    \Vert \ue^{-\beta (H_\Lambda- \mu_\Lambda N_\Lambda)} \ue^{\mathrm{i} H_{\Lambda} z} A \ue^{-\mathrm{i} H_{\Lambda} z} B \Vert_1 \leq \Vert \ue^{-(\beta - s) ( H_{\Lambda} - \mu_{\Lambda} N_{\Lambda} ) } \Vert_p \Vert \ue^{-s ( H_{\Lambda} - \mu_{\Lambda} N_{\Lambda} ) } \Vert_q \Vert A \Vert \Vert B \Vert.
    \label{eq:finiteVolumeKMS1}
 \end{equation}
To obtain \eqref{eq:finiteVolumeKMS1} we also used that $H_{\Lambda}, A$, and $B$ commute with $N_{\Lambda}$. When we choose $p = \beta/(\beta-s)$ and $q = \beta/s$, the product of the first two terms on the right-hand side equals $Z_{\Lambda} < +\infty$, and hence $|F(z)| \leq \Vert A \Vert \Vert B \Vert$. Moreover, for $t \in \mathbb{R}$ we have $F(t) = \gamma_{\Lambda}(\tau_t(A) B)$ and $F(t-\mathrm{i}\beta) = \gamma_{\Lambda}(B \tau_t(A))$. To obtain the second equality, we used the cyclicity of the trace and again that $N_{\Lambda}$ commutes with the other operators. It remains to show that $F$ is a continuous function on $S_{\beta}$, which is analytic in the interior.

To prove this, we first use the (simultaneous) spectral decomposition 
\begin{equation}
    H_{\Lambda} = \sum_{j=1}^{\infty} E_{j} | \psi_{j} \rangle \langle \psi_{j} |, \qquad N_{\Lambda} = \sum_{j=1}^{\infty}  n_j | \psi_{j} \rangle \langle \psi_{j} |
\end{equation}
to write $F$ as
\begin{equation}
    F(z) = Z_{\Lambda}^{-1} \sum_{j=1}^{\infty} \ue^{-E_{j}(\beta - \mathrm{i} z ) + \beta \mu_{\Lambda} n_{j}} \sum_{k=1}^{\infty} \ue^{-\mathrm{i} E_{k} z}  \langle \psi_{j}, A \psi_{k} \rangle \langle \psi_{k},  B \psi_{j} \rangle.
    \label{eq:finiteVolumeKMS2}
\end{equation}
The summands in the above equation are entire functions and we claim that the sum converges uniformly for $z \in S_{\beta}$. We choose $p = \beta/(\beta+\mathrm{Im}(z))$ and $q = -\beta/\mathrm{Im}(z)$ and estimate
\begin{align}
    &|\ue^{-E_{j}(\beta - \mathrm{i} z ) + \beta \mu_{\Lambda} n_j} \ue^{-\mathrm{i} E_{k} z} \langle \psi_{j}, A \psi_{k} \rangle \langle \psi_{k},  B \psi_{j} \rangle | \label{eq:finiteVolumeKMS3} \\
    &\hspace{1cm}\leq \ue^{-E_{j}(\beta + \mathrm{Im}( z ) )  + \beta \mu_{\Lambda} n_j} \ue^{\mathrm{Im}(z)  E_{k}} \frac{1}{2} \left( |\langle \psi_{j}, A \psi_{k} \rangle|^2 + | \langle \psi_{k},  B \psi_{j} \rangle |^2 \right) \notag \\
    &\hspace{1cm}\leq \left( \frac{1}{p} \ue^{-\beta ( E_{j} - \mu_{\Lambda} n_j )} + \frac{1}{q} \ue^{-\beta ( E_{k} - \mu_{\Lambda} n_{k} )} \right)  \frac{1}{2} \left( |\langle \psi_{j}, A \psi_{k} \rangle|^2 + | \langle \psi_{k},  B \psi_{j} \rangle |^2 \right). \notag
\end{align}
To obtain this result, we used again that $A$ and $B$ commute with $N_{\Lambda}$ to replace $n_j$ by $n_k$. The right-hand side satisfies the bound
\begin{align}
    &Z_{\Lambda}^{-1} \sum_{j=1}^{\infty} \sum_{k=1}^{\infty} \left( \frac{1}{p} \ue^{-\beta ( E_{j} - \mu_{\Lambda} n_{j} )} + \frac{1}{q} \ue^{-\beta ( E_{k} - \mu_{\Lambda} n_{k} )} \right)  \frac{1}{2} \left( |\langle \psi_{j}, A \psi_{k} \rangle|^2 + | \langle \psi_{k},  B \psi_{j} \rangle |^2 \right) \label{eq:finiteVolumeKMS4} \\
    &\hspace{3cm}\leq \frac{1}{2} \left( \Vert A \Vert^2 + \Vert B \Vert^2 \right), \notag
\end{align}
which proves the claim. We conclude that $F$ is a continuous function on $S_{\beta}$ that is analytic in the interior and the claim of the lemma is proved.
\end{proof}

In the second lemma we derive a bound for the derivative of the time-dependent Green function related to $\gamma_{\Lambda}$ that does not depend on $\Lambda$.

\begin{lm}
\label{lem:boundDerivativeGreensFunction}
    Assume the hypothesis and notation of Theorem~\ref{thm:KMS}.
    For $X\Subset \Gamma$, $T>0$ there exists a constant $C > 0$ so that for all $A \in \mathcal{R}^\mathrm{inv}_X$, $B \in \mathcal{R}^\mathrm{inv}$, $R>r$ and $|t|\leq T$
     \begin{equation*}
         \Big| \frac{\mathrm{d}}{\mathrm{d}t} \gamma_{\Lambda} \big(\tau_t^{X[R]}(A) B) \Big| + \Big| \frac{\mathrm{d}}{\mathrm{d}t} \gamma_{\Lambda} (B \tau_t^{X[R]}(A)) \Big| \leq C \Vert A \Vert \Vert B \Vert.
     \end{equation*}
\end{lm}
\begin{proof}
   We have
    \begin{equation}
        \frac{\mathrm{d}}{\mathrm{d}t} \gamma_\Lambda(\tau_t^{X[R]}(A)B)=\ui \gamma_\Lambda( \tau^{X[R]}_t( [H_{X[r]}, A] ) B )
    \end{equation}
    since $A\in  \mathcal{R}^\mathrm{inv}_X$ and the range of $v$ is $r$.
    Moreover,
    \begin{equation}
        | \gamma_\Lambda( \tau^{X[R]}_t( [H_{X[r]}, A] ) B ) |
        \leq \Vert B \Vert \sqrt{ \gamma_\Lambda\big(\tau^{X[R]}_t([H_{X[r]},A] [H_{X[r]},A]^*) \big) }.
        \label{eq:derivativeBound3}
    \end{equation}
    When multiplied out, the product of the two commutators reads
    \begin{align}
        &H_{X[r]} A A^* H_{X[r]} + A H_{X[r]}^2 A^* - H_{X[r]} A H_{X[r]} A^* - A H_{X[r]} A^* H_{X[r]} \label{eq:derivativeBound4} \\
        &\hspace{2cm} \leq 2 ( H_{X[r]} A A^* H_{X[r]} + A H_{X[r]}^2 A^* ) \leq 2 ( \Vert A \Vert^2 H_{X[r]}^2 + A H_{X[r]}^2 A^*). \notag
    \end{align}
    The first bound follows from the Cauchy--Schwarz inequality. 
    It is not difficult to check that there exists a constant $C(|X|,r,\sigma)$ such that
    \begin{equation}
        H_{X[r]}^2 \leq C ( 1 + N_{X[r]}^4 ) 
        \label{eq:derivativeBound5}
    \end{equation}
    holds. Using that $[A,N_{X[r]}]=[A,N] = 0$, and \eqref{eq:derivativeBound5}, we also find
    \begin{equation}
        A H_{X[r]}^2 A^* \leq C A ( 1 + N_{X[r]}^4 ) A^* = C (1+ N_{X[r]}^4)^{1/2} A A^* (1+ N_{X[r]}^4)^{1/2} \leq C \Vert A \Vert^2 ( 1 + N_{X[r]}^4 ).
        \label{eq:derivativeBound6}
    \end{equation}
    In combination with the fact that $ N_{X[r]}^4 \leq |X[r]|^3 \sum_{x \in X[r]} N_x^4$, the bounds  \eqref{eq:derivativeBound3}--\eqref{eq:derivativeBound6} and the assumption $p > 2d+2 \geq 4$ imply
    \begin{equation}
        | \gamma_\Lambda( \tau^{X[R]}_t( [H_{X[r]}, A] ) B ) | \leq C(|X|,\sigma,r) \Vert A \Vert \Vert B \Vert \sup_{x\in X[r]}\sqrt{\gamma_\Lambda(\tau^{X[R]}_t(1+N_x^4))}.
        \label{eq:derivativeBound8}
    \end{equation}
    The right hand side is bounded uniformly in $\Lambda$ and $|t|\leq T$ by the hypothesis and Proposition~\ref{prop:propagate_N}.
    The bound on the derivative of $\gamma_\Lambda(B\tau_t^{X[R]}(A))$ follows from a similar argument (interchanging the roles of $A$ and $A^*$ in all formulas starting from \eqref{eq:derivativeBound4}), and this proves the claim.
\end{proof}

With this, we are prepared to prove Theorem~\ref{thm:KMS}.

\begin{proof}[Proof of Theorem~\ref{thm:KMS}]
We apply the Banach--Alaoglu theorem to obtain a subnet $\{ \gamma_{\Lambda_\alpha} \}_{\alpha \in I}$ of the net $\{ \gamma_{\Lambda} \}_{\Lambda \Subset \Gamma}$ (with the subsets $\Lambda\Subset \Gamma$ ordered by inclusion) converging to a state $\gamma$ on $\mathcal{B}$. The uniform moment bound clearly passes to the limit, so $\gamma$, $\gamma_\Lambda$ satisfy the hypothesis~\eqref{eq:moment_bound}. From Lemma~\ref{lem:Lambda-KMS} we know that the finite-volume Gibbs states $\rho_{\Lambda}$ are $(\tau^\Lambda,\beta)$--KMS states on the local algebra $\mathcal{R}^\mathrm{inv}_\Lambda$. For given $X \subset \Lambda \Subset \Gamma$ and two operators $A, B \in \mathcal{R}^\mathrm{inv}_X$, there therefore exists an analytic function $F_{\Lambda} :  S_{\beta} \to \mathbb{C}$, which satisfies $F_{\Lambda}(t) = \gamma_{\Lambda}(\tau^{\Lambda}_t(A)B)$ and $F_{\Lambda}(t-\mathrm{i}\beta) = \gamma_{\Lambda}( B \tau^{\Lambda}_t(A))$ for all $t \in \mathbb{R}$. By the maximum principle, see e.g.~\cite[Prop.~IV.4.3]{Simon1993}, $F_\Lambda$ attains its maximum on the boundary, and since $|F_\Lambda(t)|, |F_\Lambda(t-\ui \beta)|\leq \|A\| \|B\|$ the family $F_\Lambda$ is bounded uniformly in $\Lambda$.
We now claim that $F_{\Lambda_\alpha}(t)$, $F_{\Lambda_\alpha}(t-\ui\beta)$ converge uniformly on compacts to $\gamma(\tau_t(A)B)$, respectively $\gamma(B\tau_t(A))$. For $X[R]\subset \Lambda_\alpha$, consider the difference
\begin{multline}
    \sup_{|t|\leq T}\big|\gamma(\tau_t(A)B)- F_{\Lambda_\alpha}(t)\big| \leq  \sup_{|t|\leq T}\big|\gamma(\tau_t^{X[R]}(A)B)-\gamma_{\Lambda_\alpha}(\tau_t^{X[R]}(A)B)\big|  \label{eq:proofTheoremKMS1} \\
    + \sup_{|t|\leq T}\big( \big|\gamma(\tau_t(A)B) - \gamma(\tau_t^{X[R]}(A)B)\big|+  \big|\gamma_{\Lambda_\alpha}(\tau^{\Lambda_\alpha}_t(A)B) - \gamma_{\Lambda_\alpha}(\tau_t^{X[R]}(A)B)\big| \big).
\end{multline}
Applying Theorem~\ref{thm:unif_approx} twice, with $\Lambda\equiv \Gamma$ and $\Lambda\equiv\Lambda_\alpha$, we can make the second line smaller than any $\eps>0$ by choosing $R$ sufficiently large. Moreover, $t\mapsto \gamma_{\Lambda_\alpha}(\tau_t^{X[R]}(A)B)$ is a net of functions that converges point-wise to $\gamma(\tau_t^{X[R]}(A)B)$ and whose derivatives are uniformly bounded by Lemma~\ref{lem:boundDerivativeGreensFunction}. Hence, by~\cite[Lem. 6.3.23]{BraRob2}, these functions converge uniformly on compact sets and thus we can make~\eqref{eq:proofTheoremKMS1} smaller than $2\eps$ by choosing first $R$ and then $\alpha$ large enough.
The same argument applies to $F_{\Lambda_\alpha}(t-\ui \beta)$.

We now define $\Lambda_n:=\Lambda_{\alpha_n}$ by choosing $\alpha_n \succ \alpha_{n-1}$ so that for $T=n$,~\eqref{eq:proofTheoremKMS1} and its analogue for $F_{\Lambda_\alpha}(t-\ui \beta)$ are less than $n^{-1}$ for all $\alpha\succ \alpha_n$.
Then, $\{ F_{\Lambda_n}\}_{n\in \N}$ is a bounded sequence of continuous functions on $S_\beta$, holomorphic in the interior, whose boundary values converge uniformly on compact intervals to $\gamma(\tau_t(A)B)$, $\gamma(B\tau_t(A))$.
 Using~\cite[Prop.~IV.4.3]{Simon1993}, we conclude that there is a continuous function $F$ on $S_{\mathrm{\beta}}$, holomorphic in the interior, with these boundary values, and $F_{\Lambda_n}\to F$ locally uniformly.
 This proves the existence of the analytic function in Definition~\ref{def:KMS} if $A,B\in \mathcal{R}^\mathrm{inv}_X$ for some $X \Subset \Gamma$.

For general $A,B\in \mathcal{R}^\mathrm{inv}$, we choose sequences $A_n,B_n\in \mathcal{R}^\mathrm{inv}_{X_n}$ with $X_n \Subset \Gamma$ of localized operators that converge to $A,B$ in norm and satisfy $\|A_n\|\leq \|A\|$, $\|B_n\|\leq \Vert B \Vert$. Then, as proved above, there are continuous functions $F_n$ on $S_\beta$, holomorphic in the interior, with boundary values $\gamma(\tau_t(A_n)B_n)$, $\gamma(B_n\tau_t(A_n))$. These functions are bounded by $\|A_n\| \|B_n\| \leq \|A\| \|B\|$ and their boundary values converge point-wise to the continuous functions $\gamma(\tau_t(A)B)$, $\gamma(B\tau_t(A))$. Thus, again by~\cite[Prop.~IV.4.3]{Simon1993}, they converge locally uniformly to a continuous function $F$, holomorphic inside $S_\beta$, with these boundary values.  

It remains to show that $\gamma(\tau_t(A)) = \gamma(A)$ holds for all $A \in \mathcal{R}^\mathrm{inv}$. But this follows from our approximations above and the fact that the same invariance holds for the finite-volume Gibbs states $\rho_{\Lambda}$. This completes the proof.
\end{proof}

Note that in this proof we used the uniformity of the local approximation from Theorem~\ref{thm:unif_approx} with respect to both, $t$ and $\gamma$. Together with Lemma~\ref{lem:boundDerivativeGreensFunction}, this allowed us to establish locally uniform convergence and extract a convergent subsequence (since this topology is first-countable).
It was necessary to pass from the net to a sequence, since~\cite[Prop.~IV.4.3]{Simon1993} relies on the Dominated Convergence Theorem, which may fail for nets. 

\section*{Acknowledgements}

The authors thank Alex Elgart for insightful comments on an earlier draft of the manuscript.
It is a pleasure for A.D. to acknowledge an inspiring correspondence with Detlev Buchholz. J. L. thanks Mathieu Lewin for inspiring discussions on the topic of this article and acknowledges funding by the French National Research Agency (ANR) through the project (MaBoP, ANR-23-CE40-0025; PI 1114/8-1) and the EIPHI Graduate School (ANR-17-EURE-0002) and Bourgogne-Franche-Comté Region through the project SQC.
A. D. gratefully acknowledges funding from the Swiss National Science Foundation (SNSF) through the Ambizione grant PZ00P2 185851. Part of this research was performed while A. D. was visiting the Institute for Pure and Applied Mathematics (IPAM) in Los Angeles, which is supported by the National Science Foundation (Grant No. DMS-1925919).
The research of M.L.\ is supported by the DFG (grants TRR 352 – Project-ID 470903074 and FOR 5413 - Project-ID 465199066) and by the European Union (ERC StG MathQuantProp, Grant Agreement 101163620).
\footnote{%
    Views and opinions expressed are however those of the authors only and do not necessarily reflect those of the European Union or the European Research Council Executive Agency.
    Neither the European Union nor the granting authority can be held responsible for them.
}

\section*{AI statement}

Large language models were used during the preparation of the final manuscript to improve the clarity of the writing and assist in checking the draft for errors and inconsistencies. The authors reviewed and verified all suggested changes and remain fully responsible for the content of the manuscript.


%

\end{document}